\documentclass[letterpaper,10pt, conference]{IEEEtran}
\IEEEoverridecommandlockouts
\usepackage{amsmath, graphicx, amsfonts, amssymb}

\usepackage{graphics} 
\usepackage{epsfig} 
\usepackage{mathptmx} 
\usepackage{times} 
\usepackage{psfrag}
\usepackage{subfigure}

\newtheorem{theorem}{Theorem}

\newtheorem{definition}{Definition}
\newtheorem{example}{Example}

\newtheorem{lemma}{Lemma}

\begin{document}
\bibliographystyle{IEEE}
\title{The Meaning of Structure in Interconnected Dynamic Systems \thanks{For correspondence, please contact \texttt{eyeung@caltech.edu} or \texttt{sean.warnick@gmail.com}. $\ddagger$ Control and Dynamical Systems, California Institute of Technology; $\dagger$ Information and Decision Algorithms Laboratories, Brigham Young University, Provo, UT; $*$ Control Group, Department of Engineering, University of Cambridge, Cambridge, UK; $\circ$ Automatic Control Laboratory, KTH School of Electrical Engineering, Stockholm, SE.  }
}
\author{E. Yeung$^\ddagger$, J. Gon\c{c}alves$^*$, H. Sandberg$^\circ$, S. Warnick$^\dagger$}

\maketitle

\begin{abstract}
Interconnected dynamic systems are a pervasive component of our modern infrastructures.  The complexity of such systems can be staggering, which motivates simplified representations for their manipulation and analysis.  This work introduces the complete computational structure of a system as a common baseline for comparing different simplified representations.  Linear systems are then used as a vehicle for comparing and contrasting distinct partial structure representations.  Such representations simplify the description of a system's complete computational structure at various levels of fidelity while retaining a full description of the system's input-output dynamic behavior.  Relationships between these various partial structure representations are detailed, and the landscape of new realization, minimality, and model reduction problems introduced by these representations is briefly surveyed.
\end{abstract}

\section{Introduction}

Structure and dynamic behavior are two of the most fundamental concepts characterizing a system.   The interplay between these concepts has been a central theme in control as far back as Black's, Bode's, or Nyquist's work on feedback amplifiers \cite{nyquist1932,black1934,bode1940}, or even as early as Maxwell's analysis {\em On Governors} in 1868 \cite{maxwell1868}. 

The key property driving such analyses is the fact that an interconnection of systems yields another system.  This property suggests a natural notion of structure, as the interconnection of systems, and focuses attention on understanding how interconnections of different systems result in varieties of dynamic behaviors.      

This idea of structure as interconnection is not only critical in the analysis of systems, but it also plays a key role in system modeling.  While black box approaches to modeling seek to duplicate the input-output behavior of a system irrespective of structure, first principles approaches to modeling use the visible interconnection structure of a system to decompose the system into smaller, fundamental components.  Constitutive relations describing each of these components are then applied, and interconnection principles linking these relations then result in a model structurally and dynamically consistent with the original system. 

While such approaches have demonstrated remarkable success in electrical and mechanical domains where system components are readily visible and physically separated from each other \cite{RazaCSM96,FowlerDAndreaCSM03,DAndrea}, application of such methods to biological, social, and other domains has been more difficult.  One reason may be that these systems do not exhibit a natural structure in the same sense that previous applications have; while components of electrical and mechanical systems are compartmentalized and solid-state, the physical relationship among components of these other systems are often much more fluid \cite{DelVecchio08,DelVecchio09}.  Perhaps for these other domains different notions of structure play the role historically occupied by the interconnection of components.

This paper explores these ideas by characterizing the complete computational structure of a system and then contrasting it with three distinct partial structure representations.  These different representations include the interconnection of subsystems and the standard idea of a transfer function matrix, but it also includes a newer concept of system structure called {\em signal structure} that appears to be especially useful for characterizing systems that are difficult to compartmentalize.  Precise relationships between these various perspectives of system structure are then provided, along with a brief discussion on their implications for various questions about realization and approximation.

\section{Complete Computational Structure}

The complete computational structure of a system characterizes the actual processes it uses to sense properties of its environment, represent and store variables internally, and affect change externally.  At the core of these processes are information retrieval issues such as the encoding, storage, and decoding of quantities that drive the system's dynamics.  Different mechanisms for handling these quantities result in different system structures.  

Mathematically, state equations, or their generalization as descriptor systems \cite{luenberger78,luenberger,luenberger77} are typically used to describe these mechanisms.  Although there may be many realizations that describe the same input-output properties of a particular system, its complete computational structure is the architecture of the {\em particular realization} fundamentally used to store state variables in memory and transform system inputs to the corresponding outputs.  In this work we will focus our attention on a class of differential algebraic systems that are equivalent to a set of ordinary differential equations in state space form; we will refer to such equations as {\em generalized state equations}.

Representing a system's complete computational structure is thus a question of graphically representing the structure implied by the equations that govern its state evolution.  In this work, rather than focusing on the specific syntax of any one particular graphical modeling language, we will draw from the standard system theoretic notions of a {\em block diagram} and a {\em signal flow graph} to conduct a concrete analysis between graphical representations of a system at various levels of abstraction.  The complete computational structure of a system, then, is the description of the system with the most refined resolution, which we will characterize as a graph derived from a particular block diagram of the generalized state equations.   

To make this concept of structure precise, we begin by considering a system $G$ with generalized state space realization 
\begin{equation}
\label{eq:basicsystem}
\begin{array}{rcl}
\dot{x}&=&f(x,w,u),\\
w&=&g(x,w,u),\\
y&=&h(x,w,u).
\end{array}
\end{equation}
Note that this system is in the form of a differential algebraic equation, although we will only consider systems with differentiation index zero, implying that  (\ref{eq:basicsystem}) is always equivalent to a standard ordinary differential or difference equation of the same order \cite{schmidt}.  Typically we may consider the system (\ref{eq:basicsystem}) to be defined over continuous time, with $t\in\mathbb R \geq 0$, and with $u\in\mathbb R^m$, $x\in\mathbb R^n$, $w\in\mathbb R^l$, $y\in\mathbb R^p$, and $\dot{x}$ taken to mean $dx/dt$.  Moreover, we restrict our attention to those functions $f$, $g$ and $h$ where solutions exist for $t\geq 0$.  Nevertheless, we could also consider discrete time systems, with $t=0,1,2,3,...$ and $\dot{x}$ in (\ref{eq:basicsystem}) taken to mean $x[t+1]$, or systems with general input, state, auxiliary, and output spaces $\cal U$, $\cal X$, $\cal W$, or $\cal Y$, respectively.  In some situations these ``spaces" may merely be sets, e.g. ${\cal X}=\{0,1\}$.  In any case, however, we will take $u\in {\cal U}^m$, $x\in {\cal X}^n$, $w\in{\cal W}^l$, and $y\in{\cal Y}^p$ so that $m$, $n$, $l$ and $p$ characterize the dimensions of the input, state, auxiliary and output vectors, respectively.

Note that the auxiliary variables, $w$, are used to characterize intermediate computation in the composition of functions.  Thus, for example, we distinguish between $f(x)=x$ and $f(x)=2(.5x)$ by computing the latter as $f(w)=2w$ and $w=g(x)=.5x$.  In this way, the auxiliary variables serve to identify stages in the computation of the state space realization (\ref{eq:basicsystem}).  Frequently we may not require any auxiliary variables in our description of the system; indeed it is the standard practice \cite{willems07} to eliminate auxiliary variables to simplify the state descriptions of systems.  Nevertheless, as we discuss structure, it will be critical to use auxiliary variables to distinguish between systems with dynamically equivalent, yet structurally distinct architectures, leading to the following definition.

\begin{definition} Given a system (\ref{eq:basicsystem}), we call the number of auxiliary variables, $l$, the {\em intricacy} of the realization.
\end{definition}

To understand the structure of (\ref{eq:basicsystem}), we need a notion of dependence of a function on its arguments.  For example, the function $f(x,y,z) = xy-x+z$ clearly depends on $z$, but it only depends on $x$ when $y\neq 1$ (or on $y$ when $x\neq 0$).  Since ``structure" refers at some level to the dependence of the system variables on each other, it is important that our notion of dependence be made clear.

\begin{definition}\label{definition:dependence} A function $f(w)$, from $l$-dimensional domain $\cal W$ to $s$-dimensional co-domain $\cal Z$, is said to {\em depend} on the $i^{th}$ variable, $w_i$, if there exist values of the other $s-1$ variables $w_j$, $j\neq i$, such that $f(w)$ is not constant over all values of $w_i$ while holding the others variables fixed.  If $s=1$, then $f(w)$ {\em depends} on $w$ if it is not constant over all values of $w$.
\end{definition}

Note that when $\partial f/\partial w_i$ is well defined, the above definition coincides with the partial derivative being non-zero for some value of the variables $w_j$.  Nevertheless, here we allow for non-differentiable functions as we explicitly characterize one notion of the structure of a state space realization.

\begin{definition} Given a system $G$ with realization (\ref{eq:basicsystem}), its {\em complete} or {\em computational structure} is a weighted directed graph $\cal C$ with vertex set $V(\cal C)$, and edge set $E(\cal C)$.  The vertex set contains $m+n+l+p$ elements, with one for each input, state, auxiliary, and output variable of the system; we label the vertices accordingly.  In particular, the vertex associated with the $i^{th}$ input is labeled $u_i$, $1\leq i\leq m$, the vertex associated with the $j^{th}$ state is labeled $f_j$, $0\le j\le n$, the vertex associated with the  $k^{th}$ auxiliary variable is labeled $g_k$, $0\le k \le l$, and the vertex associated with the $r^{th}$ output is labeled $h_r$, $1\le r\leq p$. The edge set contains an edge from node $i$ to node $j$ if the function associated with the label of node $j$ depends on the variable produced by node $i$.  Moreover, the edge $(i,j)$ is then labeled (weighted) with the variable produced by node $i$.
\end{definition}

So, for example, consider the following continuous time system with real-valued variables:
\begin{equation}
\label{eq:simple_example}
\begin{array}{ccl}
\left[\begin{array}{c}\dot{x}_1\\\dot{x}_2\end{array}\right]&=&\left[\begin{array}{l}f_1(x_1,w, u)\\f_2(x_1)\end{array}\right],\\
w&=&\begin{array}{l}g(x_2)\end{array}\\
y&=&\begin{array}{l}h(x_2)\end{array}.
\end{array}
\end{equation}
Its complete computational structure has node set $V({\cal C})=\{u,f_1,f_2,g,h\}$ and edge set $E({\cal C})=\{u(u,f_1), x_1(f_1,f_1), x_1(f_1,f_2),x_2(f_2,g),x_2(f_2,h),w(g,f_1)\}$; Figure \ref{fig:simple_example} illustrates the graph.  

\begin{figure}[htbp]
\centering
\subfigure[The complete computational structure $\cal C$ of the simple example specified by equation (\ref{eq:simple_example}).]{
          \includegraphics*[width=.45\textwidth, viewport = 130 200 650 415]{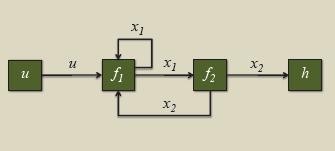}
          \label{fig:simple_example_a}
          }\\
     \subfigure[A modified representation of the complete computational structure of the system specified by equation (\ref{eq:simple_example}).  Since the edges leaving a particular node will always represent the same variable, they have been combined to simplify the figure.]{
          \includegraphics*[width=.45\textwidth, viewport = 130 200 650 415]{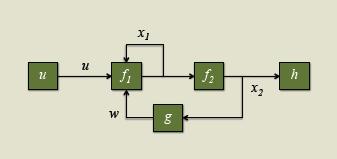}
          \label{fig:simple_example_b}
          }
\caption{The complete computational structure of the state realization of a system is a graph demonstrating the dependency among all system variables; edges correspond to variables and nodes represent constitutive mechanisms that produce each variable.  System outputs are understood to leave the corresponding terminal nodes, and system inputs arrive at the corresponding source nodes.}
\label{fig:simple_example}
\end{figure}

This notion of the complete computational structure of the system (\ref{eq:basicsystem}) corresponds with the traditional idea of the structure of a dynamic system (see for example \cite{Siljac,}), but with some important differences.  First, this description uses auxiliary variables to keep track of the structural differences introduced by the composition of functions. This allows a degree of flexibility in how refined a view of the computational structure one considers ``complete."  Also, this description may have some slight differences in the labeling of nodes and edges compared with various descriptions in the literature \cite{Siljac,murrayCON,roy09,Parlangeli,willems07}.  These differences will become important as new notions of partial structure are introduced in later sections.  Here, we use rectangular nodes for graphs where the nodes represent {\em systems}, and the associated edges will represent {\em signals}.  This convention will bridge well between the graphical representation of system structure and the typical engineering block diagram of a system, and it sometimes motivates the simplification in drawing edges as shown in Figure \ref{fig:simple_example_b}), since every edge that leaves a node represents the same variable and carries the same label.  Moreover, notice that nodes associated with the mechanisms that produce output variables are {\em terminal}, in that no edge ever leaves these nodes, while the nodes associated with input variables are {\em sources}, in that no edge ever arrives at these nodes.  Although it is common for engineering diagrams to explicitly draw the edges associated with output variables and leave them ``dangling," with no explicit terminal node, or to eliminate the input nodes and simply depict the input edges--also ``dangling," our convention ensures that the diagram corresponds to a well defined {\em graph}, with every edge characterized by an ordered pair of nodes.  Note also that state nodes, such as $f_1$ in Figure \ref{fig:simple_example}, may have self loops, although auxiliary nodes will not, and at times it will be convenient to partition the vertex set into groups corresponding to the input, state, auxiliary, and output mechanisms as $V({\cal C}) = \{V_u({\cal C}), V_x({\cal C}), V_w({\cal C}), V_y({\cal C})\}$.  Likewise, we may similarly partition the edge set as necessary.

We see, then, that knowing the complete structure $\cal C$ of a system is equivalent to knowing its state space realization, along with the composition structure with which these functions are represented, given by $(f,g,h)$.  We refer to this structure as {\em computational} because it reveals the dependencies among variables in the particular representation, or basis, that they are stored in and retrieved from memory.  These specific, physical mechanisms that store and retrieve information, identified with $V_x({\cal C})$, are interconnected with devices that transform variables, identified with $V_w({\cal C})$, and with devices that interface with the system's external environment.  These devices include sensors, identified with $V_u({\cal C})$, and actuators, identified with $V_y({\cal C})$, to implement the particular system behavior observed by the outside world through the {\em manifest} variables, $u$ and $y$.  Although other technologies very well may implement the same observed behavior via a different computational structure and a different representation of the {\em hidden} variables, $x$ and $w$, $\cal C$ describes the structure of the actual system employing existing technologies as captured through a particular state description.  In this sense, $\cal C$ is the complete architecture of the system, and often may be interpreted as the system's ``physical layer."  Importantly, it is often this notion of structure, or a very related concept, that is meant when discussing the ``structure" of a system, as the next example illustrates.  

\subsection{Example: Graph Dynamical Systems}

As an example, we examine the computational structure of a graph dynamical system (GDS).  Graph dynamical systems are finite alphabet, discrete time systems with dynamics defined in terms of the structure of an associated undirected graph.  They have been employed in the study of various complex systems \cite{Mortveit01,Mortveit2000}, including
\begin{itemize}
\item Dynamical process on networks:
\begin{itemize}
\item disease propagation over a social contact graph, 
\item packet flow in cell phone communication, 
\item urban traffic and transportation;
\end{itemize}
\item Computational algorithms:
\begin{itemize}
\item Gauss-Seidel, 
\item gene annotation based on functional linkage networks, 
\item transport computations on irregular grids;
\end{itemize}
\item Computational paradigms related to distributed computing.
\end{itemize}
Here we observe that the computational structure of the graph dynamical system corresponds naturally with the system's underlying graph. 

Given an undirected graph $\cal G$ with vertex set $V({\cal G})=\{1,2,...,n\}$, a GDS associates with each node $i$ a state $x_i$ that takes its values from a specified finite set $\cal X$.  This state is then assigned a particular update function that updates its value according to the values of states associated with nodes adjacent to node $i$ on $\cal G$.  Notice that this restriction on the update function suggests that the update function for state $i$ {\em depends} on states consistent with the structure of $\cal G$, indicating that the system's computational structure $\cal C$ should correspond to the adjacency structure of $\cal G$.

The distinction is made between a GDS that updates all of its states simultaneously, called a {\em parallel GDS}, and one that updates its states in a particular sequence, called a {\em sequential GDS}.  The parallel GDS thus becomes an autonomous dynamical system, evolving its state according to its update function from some initial condition.  The sequential GDS, on the other hand, can be viewed as a controlled system that receives a particular permutation of the node set $V({\cal G})$ as input and then evolves its states accordingly. 
\begin{figure}[htbp]
\centering
\subfigure[Undirected graph $\cal G$ defining the sequential Graph Dynamical System (\ref{eq:GDSexample}).]
					{\includegraphics*[width=.45\textwidth, viewport = 250 230 540 380]{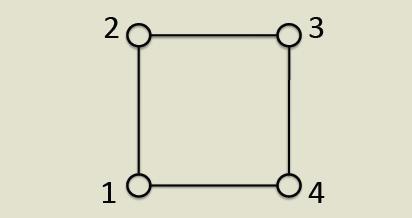}
					\label{fig:GDSa}}
\subfigure[Computational structure $\cal C$ of the sequential Graph Dynamical System (\ref{eq:GDSexample}).]
       {\includegraphics*[width=.45\textwidth, viewport = 120 140 670 470]{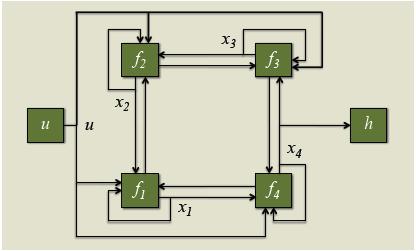}
          \label{fig:GDSb}}
 \caption{Graph Dynamical Systems are finite alphabet discrete time systems with dynamics defined in terms of the adjacency structure of a specified undirected graph.  Here we note that the computational structure of the system reflects the structure of its underlying graph. }
\label{fig:GDS}
\end{figure}

To illustrate, consider the sequential GDS given by ${\cal G}=\{1,2,3,4\}$ as indicated in Figure \ref{fig:GDS}.  Let ${\cal U}=\{1,2,3,4\}$ and ${\cal X}={\cal Y} = \{0,1\}$, with update function given by
\begin{equation}
\label{eq:GDSexample}
\begin{array}{ccc}
\left[\begin{array}{c}x_1[t+1]\\x_2[t+1]\\x_3[t+1]\\x_4[t+1]\end{array}\right] & = &\left[\begin{array}{c}f_1(x[t],u[t])\\f_2(x[t],u[t])\\f_3(x[t],u[t])\\f_4(x[t],u[t])\end{array}\right]\\\\
y[t] & = & x_4[t]
\end{array}
\end{equation}
where
\begin{equation}
\label{eq:update}
f_i(x,u) = \left\{\begin{array}{ll} x_i & u\neq i\\ (1+x_i)(1+x_{i-1})(1+x_{i+1})&u=i\end{array}\right.
\end{equation}
and the arithmetic in (\ref{eq:update}) is taken modulo 2, while that in the subscript notation is taken modulo 4 (resulting in $x_5\equiv x_1$, etc.).  So, for example, the initial condition $x[0] = [0\;0\;0\;0]^T$ with input sequence \hbox{$u[t] = 1, 2, 3, 4, 1, 2, 3, 4,...$} would result in the following periodic trajectory:
\begin{equation*}
\begin{array}{cc}
\begin{array}{c}
x[1] = [1\;0\;0\;0]^T\\
\downarrow\\
x[2] = [1\;0\;0\;0]^T\\
\downarrow\\
x[3] = [1\;0\;1\;0]^T\\
\downarrow\\
x[4] = [1\;0\;1\;0]^T\\
\downarrow\\
x[8] = [0\;0\;0\;1]^T
\end{array}
&
\begin{array}{c}
x[12] = [0\;1\;0\;0]^T\\
\downarrow\\
x[16] = [0\;0\;1\;0]^T\\
\downarrow\\
x[20] = [1\;0\;0\;0]^T\\
\downarrow\\
x[24] = [0\;1\;0\;1]^T\\
\downarrow\\
x[28] = [0\;0\;0\;0]^T
\end{array}
\end{array}
\end{equation*}

The computational structure of the system (\ref{eq:GDSexample}) follows immediately from the dependency among variables characterized by equation (\ref{eq:update});  Figure \ref{fig:GDS} illustrates $\cal C$ for this system.  Notice that the structure of $\cal G$ is reflected in $\cal C$, where the undirected edges in $\cal G$ have been replaced by directed edges in both directions, and self-loops have appeared where applicable.  Moreover, notice that $\cal C$ considers the explicit influence of the input sequence on the update computation, and explicitly identifies the system output.  In this way, we see that the computational structure is a reflection of the natural structure of the sequential GDS.

\subsection{Computational Structure of Linear Systems}
Linear systems represent an important special case of those described by (\ref{eq:basicsystem}).  They arise naturally as the linearization of sufficiently smooth nonlinear dynamics near an equilibrium point or limit cycle, or as the fundamental dynamics of systems engineered to behave linearly under nominal operating conditions.  In either case, knowing the structure of the relevant linear system is a critical first step to understanding that of the underlying nonlinear phenomena.

The general state description of a linear system is given by
\begin{equation}
\label{eq:linearsystem}
\begin{array}{ccl}
\dot{x}&=&Ax+\hat{A}w+Bu,\\
w&=&\bar{A}x+\tilde{A}w+\bar{B}u,\\
y &=&Cx+\bar{C}w+Du,
\end{array}
\end{equation}
where $A\in \mathbb R^{n\times n}$, $\hat{A}\in \mathbb R^{n\times l}$, $\bar{A}\in \mathbb R^{l\times n}$, $\tilde{A}\in \mathbb R^{l\times l}$, $B\in \mathbb R^{n\times m}$, $\bar{B}\in \mathbb R^{l\times m}$, $C\in \mathbb R^{p\times n}$, $\bar{C}\in \mathbb R^{p\times l}$, and $D\in \mathbb R^{p\times m}$.  Note that $I-\tilde{A}$ is necessarily invertible, ensuring that the differentiability index of the system is zero.  Nevertheless, the matrices are otherwise free.  

As in the nonlinear case, it should be apparent that the auxiliary variables are superfluous in terms of characterizing the dynamic behavior of the system; this idea is made precise in the following lemma.  Nevertheless, the auxiliary variables make a very important difference in terms of characterizing the system's complete computational structure, as illustrated by the subsequent example.  

\begin{lemma}  For any system (\ref{eq:linearsystem}) with intricacy $l>0$, there exists a unique {\em minimal intricacy} realization $(A_o,B_o,C_o,D_o)$ with $l=0$ such that for every solution $(u(t),x(t),w(t),y(t))$ of ($\ref{eq:linearsystem}$), $(u(t),x(t),y(t))$ is  a solution of $(A_o,B_o,C_o,D_o)$.
\label{minintrlemma}
\end{lemma}

\begin{proof} The result follows from the  invertibility of $(I-\tilde{A})$.  Solving for $w$ and substituting into the equations of $\dot{x}$ and $y$ then yields $(A_o,B_o,C_o,D_o)$. 
\end{proof}

Consider, for example, the system (\ref{eq:linearsystem}) with state matrices given by $D=0$ and the following:
\begin{equation}
\label{eq:example_l8}
\begin{array}{lr}
A = \left[\begin{array}{cc}0&0\\0&0\end{array}\right] & \hat{A}=\left[\begin{array}{cccccccc}0&0&1&0&1&0&1&0\\0&0&0&1&0&1&0&1\end{array}\right]\\\\
\bar{A} = \left[\begin{array}{cc}c_1&0\\0&c_2\\0&0\\0&0\\0&0\\0&0\\a_1&0\\0&a_2\end{array}\right]&\tilde{A}=\left[\begin{array}{ccccccccc}0&0&0&0&0&0&0&0\\0&0&0&0&0&0&0&0\\0&e_1&0&0&0&0&0&0\\e_2&0&0&0&0&0&0&0\\0&0&0&0&0&0&0&0\\0&0&0&0&0&0&0&0\\0&0&0&0&0&0&0&0\\0&0&0&0&0&0&0&0\\\end{array}\right]\\\\
B = \left[\begin{array}{cc}0&0\\0&0\end{array}\right] & \bar{B}^T=\left[\begin{array}{cccccccc}0&0&0&0&b_1&0&0&0\\0&0&0&0&0&b_2&0&0\end{array}\right]\\\\
C = \left[\begin{array}{cc}0&0\\0&0\end{array}\right] &\bar{C} = \left[\begin{array}{cccccccc}1&0&0&0&0&0&0&0\\0&1&0&0&0&0&0&0\end{array}\right]\\\\
\end{array}
\end{equation}
This system has the complete computational structure $\cal C$ shown in Figure \ref{fig:linear_example_a}.  Here, because each auxiliary variable is defined as the simple product of a coefficient times another variable, we label the node corresponding to $w_i$ in $\cal C$ with the appropriate coefficient rather than the generic label, $g_i$.  Note that this realization has an intricacy of $l=8$.  

Suppose, however, that we eliminate the last six auxiliary variables, leading to an equivalent realization with intricacy $l=2$.  The state matrices then become
\begin{equation}
\label{eq:example_l2}
\begin{array}{lr}
A = \left[\begin{array}{cc}a_1&0\\0&a_2\end{array}\right] & \hat{A}=\left[\begin{array}{cc}0&e_1\\e_2&0\end{array}\right]\\\\
\bar{A} = \left[\begin{array}{cc}c_1&0\\0&c_2\end{array}\right]&\tilde{A}=\left[\begin{array}{cc}0&0\\0&0\end{array}\right]\\\\
B = \left[\begin{array}{cc}b_1&0\\0&b_2\end{array}\right] & \bar{B}^T=\left[\begin{array}{cc}0&0\\0&0\end{array}\right]\\\\
C = \left[\begin{array}{cc}0&0\\0&0\end{array}\right] &\bar{C} = \left[\begin{array}{cc}1&0\\0&1\end{array}\right]
\end{array}
\end{equation}
with computational structure $\cal C$ as shown in Figure \ref{fig:linear_example_b}.  Similarly, we can find an equivalent realization with $l=0$ given by
\begin{equation}
\label{eq:example_l0}
\begin{array}{ccc}
A_o = \left[\begin{array}{cc}a_1&e_1c_2\\e_2c_1&a_2\end{array}\right] & 
B_o = \left[\begin{array}{cc}b_1&0\\0&b_2\end{array}\right] & 
C_o = \left[\begin{array}{cc}c_1&0\\0&c_2\end{array}\right]\\\\ 
\end{array}
\end{equation}
and all other system matrices equal to zero.  This realization, (\ref{eq:example_l0}), is the minimal intricacy realization of both systems (\ref{eq:example_l8}) and (\ref{eq:example_l2}), and its complete computational structure $\cal C$ is given in Figure \ref{fig:linear_example_c}.  The equivalence between these realizations is easily verified by substitution.

\begin{figure}[htb]
\centering
\subfigure[The complete computational structure $\cal C$ of the linear system given by the state matrices (\ref{eq:example_l8}) with intricacy $l=8$.]{
          \includegraphics*[width=.45\textwidth, viewport = 110 80 680 530 ]{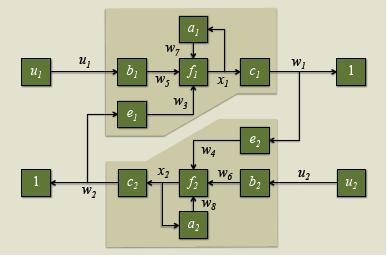}
          \label{fig:linear_example_a}
          }\\
     \subfigure[The complete computational structure $\cal C$ of the equivalent linear system with intricacy $l=2$, specified by the realization (\ref{eq:example_l2}). ]{
          \includegraphics*[width=.45\textwidth, viewport = 110 150 680 460]{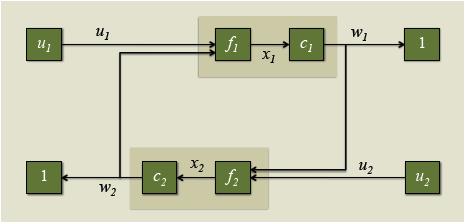} 
          \label{fig:linear_example_b}
          }\\
         \subfigure[The complete computational structure $\cal C$ of the minimal intricacy \hbox{($l=0$)} realization for both systems above, characterized by (\ref{eq:example_l0}).]{
          \includegraphics*[width=.45\textwidth, viewport = 110 150 680 460]{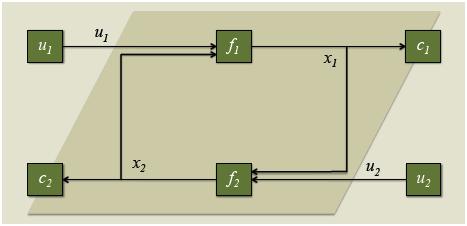}
          \label{fig:linear_example_c}
           }
\caption{The complete computational structure of a linear system characterized by realizations of differing intricacies.  Edges within shaded regions represent {\em hidden} variables, while those outside shaded regions are {\em manifest} variables.}
\label{fig:linear_example}
\end{figure}

Comparing the computational structures for different realizations of the same system, (\ref{eq:example_l8}), (\ref{eq:example_l2}), and (\ref{eq:example_l0}), we note that the intricacy of auxiliary variables plays a critical role in suppressing or revealing system structure.  Moreover, note that auxiliary variables can change the nature of which variables are manifest or hidden.  In Figure \ref{fig:linear_example} shaded regions indicate which variables, represented by edges, are hidden; manifest variables leave shaded regions while hidden variables are contained within them.  Note that the minimal intricacy realization has no internal manifest variables, or, in other words, it has  a single block of hidden variables (Figure \ref{fig:linear_example_c}).   Meanwhile, both $w_1$ and $w_2$ are manifest in the other realizations (Figures \ref{fig:linear_example_a}, \ref{fig:linear_example_b}) since $w_1=y_1$ and $w_2=y_2$, indicated by the ``1" at their respective terminal nodes.  This yields two distinct blocks of hidden variables, in either case, revealing the role intricacy of a realization can play characterizing its structure.

The complete computational structure of a system is thus a graphical representation of the dependency among input, state, auxiliary, and output variables that is in direct, one-to-one correspondence with the system's state realization, generalized to explicitly account for composition intricacy.  All structural and behavioral information is fully represented by this description of a system.  Nevertheless, this representation of the system can also be unwieldy for large systems with intricate structure.  

\section{Partial Structure Representations}

Complex systems are often characterized by intricate computational structure and complicated dynamic behavior.  State descriptions and their corresponding complete computational structures accurately capture both the system's structural and dynamic complexity, nevertheless these descriptions themselves can be too complicated to convey  an efficient understanding of the nature of the system.  Simplified representations are then desirable.    

One way to simplify the representation of a system is to restrict the structural information of the representation while maintaining a complete description of the system's dynamics.  The most extreme example of this type of simplified representation is the transfer function of a single-input single-output linear time invariant (LTI) system.   A transfer function completely specifies the system's input-output dynamics without retaining any information about the computational structure of the system.  For example, consider the $n^{th}$ order LTI single-input single-output system given by $(A,b,c,d)$.  It is well known that although the state description of the system completely specifies the transfer function, $G(s) = c(sI-A)^{-1}b+d$, the transfer function $G(s)$ has an infinite variety of state realizations, and hence computational structures, that all characterize the same input-output behavior.  That is, the structural information in any state realization of the system is completely removed in the transfer function representation of the system, even though the dynamic (or behavioral) information about the system is preserved.

We use this power of a transfer function to obfuscate structural information to develop three distinct partial-structure representations of an LTI system: subsystem structure, signal structure, and the sparsity pattern of a (multiple input, multiple output) system's transfer function matrix.  Later we will show how each of these representations contain different kinds of structural information, and we will precisely characterize the relationships among them.  

\subsection{Subsystem Structure}

One of the most natural ways to reduce the structural information in a system's representation is to partition the nodes of its computational structure into subsystems, then replace these subsystems with their associated transfer function. Each transfer function obfuscates the structure of its associated subsystem, and the remaining (partial) structural information in the system is the interconnection between transfer functions.

Subsystem structure refers to the appropriate decomposition of a system into constituent subsystems and the interconnection structure between these subsystems.  Abstractly, it is the condensation graph of the complete computational structure graph, $\cal C$, taken with respect to a particular partition of $\cal C$ that identifies subsystems in the system.  Such abstractions have been used in various ways to simplify the structural descriptions of complex systems \cite{Siljac,harary}, for example by ``condensing" strongly connected components or other groups of vertices of a graph into single nodes.  Nevertheless, in this work we define a {\em particular} condensation graph as the subsystem structure of the system.  We begin by characterizing the partitions of $\cal C$ that identify subsystems.

\begin{definition}\label{admissiblepartition} 
Given a system $G$ with realization (\ref{eq:linearsystem}) and associated computational structure $\cal C$, we say a partition of $V({\cal C})$ is {\em admissible} if every edge in $E({\cal C})$ between components of the partition represents a variable that is manifest, not hidden.
\end{definition}

For example, considering the system (\ref{eq:example_l0}) with $V({\cal C})=\{u_1, f_1, c_1, c_2, f_2, u_2\}$.  We see that the partition $\{(u_1),(f_1, c_1, c_2, f_2), (u_2)\}$ is admissible since the only edges between components are $u_1(u_1,f_1)$ and $u_2(u_2,f_2)$, representing the manifest variables $u_1$ and $u_2$.  Notice that the shading in Figure \ref{fig:linear_example_c} is consistent with this admissible partition.   Alternatively, the partition $\{(u_1),(f_1, c_1), (c_2, f_2), (u_2)\}$ is not admissible for (\ref{eq:example_l0}), since the edges $x_1(f_1,f_2)$ and $x_2(f_2,f_1)$ extend between components of the partition but represent variables $x_1$ and $x_2$ that are hidden, not manifest.  

Although sometimes any aggregation, or set of fundamental computational mechanisms represented by vertices in $\cal C$, may be considered a valid subsystem, in this work a subsystem has a specific meaning.  In particular, the variables that interconnect subsystems must be manifest, and thus subsystems are identified by the components of admissible partitions of $V({\cal C})$.  We adopt this convention to 1) enable the distinction between real subsystems vs. merely arbitrary aggregations of the components of a system, and 2) ensure that the actual subsystem architecture of a particular system is adequately reflected in the system's computational structure and associated realization, thereby ensuring that such realization is complete.

\begin{definition} Given a system $G$ with realization (\ref{eq:linearsystem}) and associated computational structure $\cal C$, the system's {\em subsystem structure} is a condensation graph $\cal S$ of $\cal C$ with vertex set $V({\cal S})$ and edge set $E({\cal S})$ given by:
\begin{itemize}
\item $V({\cal S})=\{S_1,...S_q\}$ are the elements of an admissible partition of $V({\cal C})$ of maximal cardinality, and
\item $E({\cal S})$ has an edge $(S_i,S_j)$ if $E({\cal C})$ has an edge from some component of $S_i$ to some component of $S_j$.
\end{itemize}
We label the nodes of $V({\cal S})$ with the transfer function of the associated subsystem, which we also denote $S_i$, and the edges of $E({\cal S})$ with the associated variable from $E({\cal C})$.
\end{definition}

Note that, like $\cal C$, the subsystem structure $\cal S$ is a graph with vertices that represent {\em systems} and edges that represent {\em signals}, or system variables.  For example, Figure \ref{fig:SSexample_a} illustrates the subsystem structure for both systems (\ref{eq:example_l8}) and (\ref{eq:example_l2}), shown in Figures \ref{fig:linear_example_a} and \ref{fig:linear_example_b}.  Note that the subsystem structure of these systems' minimally intricate realization, (\ref{eq:example_l0}), is quite different, with a single block rather than two blocks interconnected in feedback, as shown in Figure \ref{fig:SSexample_b}.  This illustrates the necessity of auxiliary variables to adequately describe the complete system structure.

\begin{figure}[htb]
\centering
\subfigure[The subsystem structure $\cal S$ of the linear system given by the state matrices (\ref{eq:example_l8}) with intricacy $l=8$.  Note that the subsystem structure of the equivalent, but less intricate system given by equation (\ref{eq:example_l2}) is exactly the same, indicated by the shaded regions in Figures \ref{fig:linear_example_a} and \ref{fig:linear_example_b}.]{
          \includegraphics*[width=.45\textwidth, viewport = 100 150 690 460]{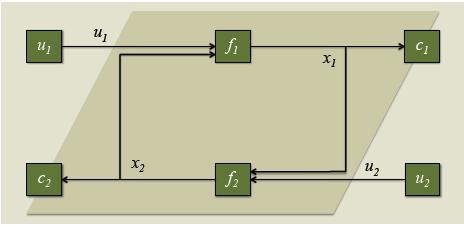} %
          \label{fig:SSexample_a}
          }\\
     \subfigure[The subsystem partial structure $\cal S$ of the dynamically equivalent, minimally intricate linear system, specified by the realization (\ref{eq:example_l0}) corresponding to Figure \ref{fig:linear_example_c}.]{
          \includegraphics*[width=.45\textwidth, viewport = 100 250 690 370]{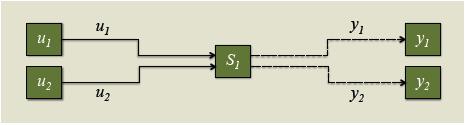}%
          \label{fig:SSexample_b}
          }
          \caption{The subsystem structure of a linear system partitions vertices of the complete computational structure and condenses admissible groups of nodes into single subsystem nodes, shown in brown; nodes shown in green are those that correspond directly to unaggregated vertices of the complete computational structure, which will always be associated with mechanisms that generate manifest variables.  Here, the subsystem structure corresponds to the interconnection of shaded regions in Figure \ref{fig:linear_example}.  Comparing Figures \ref{fig:SSexample_a} and \ref{fig:SSexample_b}, we note that intricacy in a realization may be necessary to characterize meaningful subsystems and yield nontrivial subsystem structure.}
\label{fig:SSexample}
\end{figure}

\begin{lemma}\label{spsunique}
The subsystem structure ${\cal S}$ of a system $G$, with complete computational structure ${\cal C}$, is unique. 
\end{lemma}
\begin{proof}
We prove by contradiction.  Suppose the subsystem structure ${\cal S}$ of $G$ is not unique.  Then there are at least two distinct subsystem structures of $G$, which we will label ${\cal S}^1$ and ${\cal S}^2$.  This implies there are two admissible partitions of $V({\cal C})$, given by $V({\cal S}^1)$ and $V({\cal S}^2)$, such that $V({\cal S}^1)\neq V({\cal S}^2)$ and with equal cardinality, $q$.  Note that by definition, $q$ must be the maximal cardinality of any admisible partition of $V({\cal C})$.  To obtain a contradiction, we will construct another admissible partition, $V({\cal S}^3)$, such that $|V({\cal S}^3)| > q.$

Consider  the following partition of $V({\cal C})$ that is a refinement of both $V({\cal S}^1)$ and $V({\cal S}^2)$:
\[ V({\cal S}^3) = \{S^3|S^3 \neq \emptyset; S^3= S_i \cap {S}_j, S_i \in V({\cal S}^1) , {S}_j \in V({\cal S}^2)\}.\] 
Since $V({\cal S}^1)\neq V({\cal S}^2)$, then $|V({\cal S}^3)| > q$, since the cardinality of the refinement must then be greater than that of $V({\cal S}^1)$ or $V({\cal S}^2)$.  Moreover, note that the partition $V({\cal S}^3)$ is admissible, since every edge of ${\cal C}$ between vertices associated with distinct components of $V({\cal S}^3)$ corresponds with an edge of either ${\cal S}^1$ or ${\cal S}^2$, which are admissible.  Thus, $V({\cal S}^3)$ is an admissible partition of $V(\cal C)$ with cardinality greater than $q$, which contradicts the assumption that ${\cal S}^1$ and ${\cal S}^2$ are both subsystem structures of $G$.
\end{proof}

The subsystem structure of a system reveals the way natural subsystems are interconnected, and it can be represented in other ways besides (but equivalent to) specifying $\cal S$.  For example, one common way to identify this kind of subsystem architecture is to write the system as the linear fractional transformation (LFT) with a block diagonal ``subsystem" component and a static ``interconnection" component (see \cite{ZDG96} for background on the LFT).  For example, the system in Figure \ref{fig:SSexample_a} can be equivalently represented by the feedback interconnection of a static system $N:{\mathbb U}\times{\mathbb W}\rightarrow{\mathbb Y}\times({\mathbb U}\times{\mathbb W})$ and a block-diagonal dynamic system $S:{\mathbb U}\times{\mathbb W}\rightarrow{\mathbb W}$ given by
\begin{equation}
\begin{array}{rl}
N = \left[\begin{array}{cc|cc}0&0&1&0\\0&0&0&1\\\hline1&0&0&0\\0&0&0&1\\0&0&1&0\\0&1&0&0\end{array}\right],&
S = \left[\begin{array}{cc}S_1&0\\0&S_2\end{array}\right],
\end{array}
\end{equation}
where
\begin{equation}
\begin{array}{lr}
S_1:\left[\begin{array}{c}u_1\\w_2\end{array}\right]\rightarrow w_1 & S_2 :\left[\begin{array}{c}w_1\\u_2\end{array}\right]\rightarrow w_2 \\
\end{array}
\end{equation}

\begin{figure}[t]
\centering
          \includegraphics*[width=.45\textwidth, viewport = 200 150 590 460]{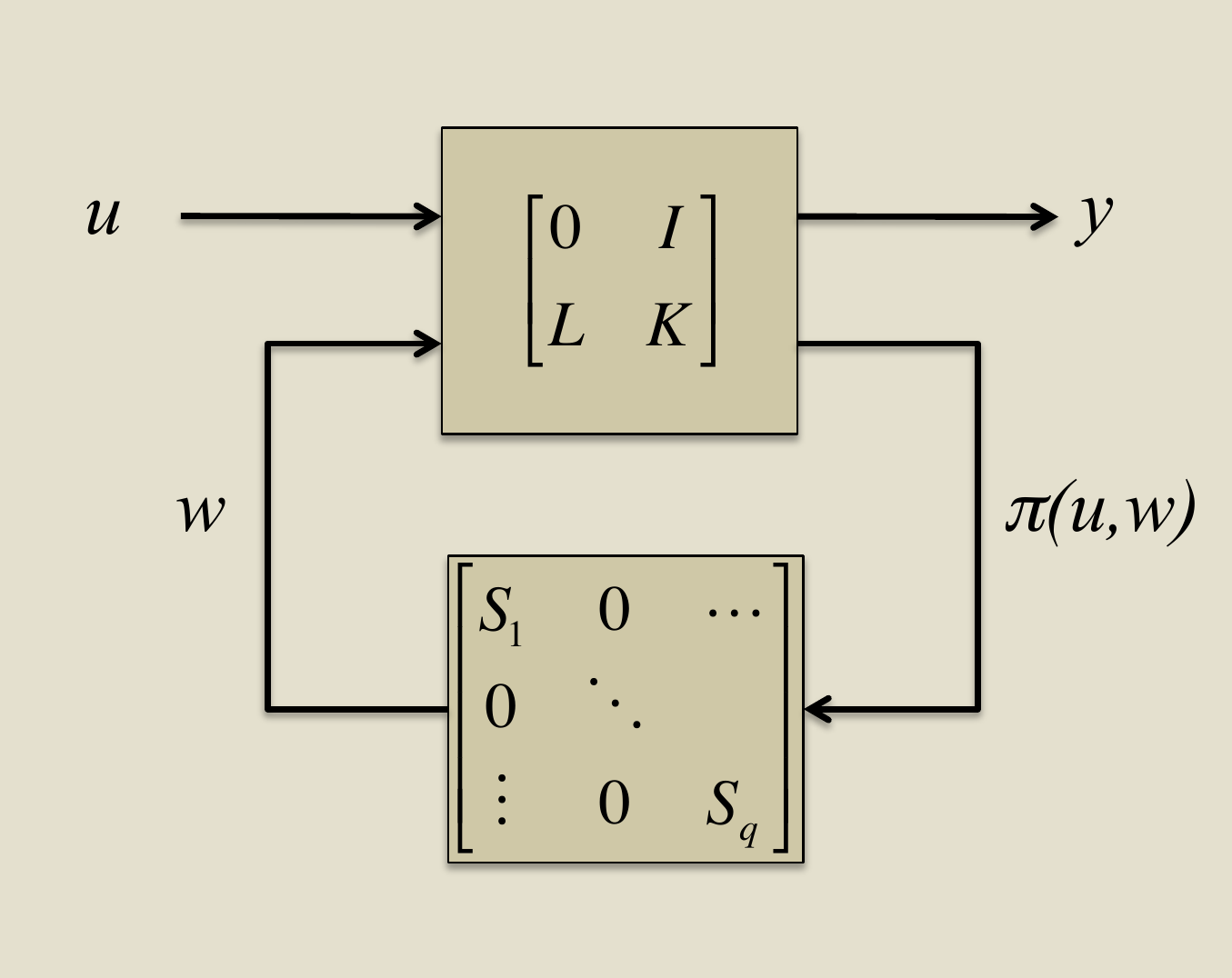} %
\caption{The subsystem structure of a system can always be represented by the lower linear fractional transformation of the static interconnection matrix $N$ with a block diagonal transfer function matrix $S$.  Note that $\pi(u,w)$ represents a permutation of a subset of the variables in the vector inputs, $u$, and manifest auxiliary variables, $w$, possibly with repetition of some variables if necessary.  }
\label{fig:LFTexample}
\end{figure}

In general, the LFT associated with $\cal S$ will have the form
\begin{equation}\label{eq:LFT}
\begin{array}{lr}
N = \left[\begin{array}{cc}0&I\\L&K\end{array}\right]&
S = \left[\begin{array}{ccc}S_1&0&...\\0&\ddots&\\\vdots&0&S_q\end{array}\right]
\end{array}
\end{equation}
where $q$ is the number of distinct subsystems, and $L$ and $K$ are each binary matrices of the appropriate dimension (see Figure \ref{fig:LFTexample}).  Note that if additional output variables are present, besides the manifest variables used to interconnect subsystems, then the structure of $N$ and $S$ above extend naturally.  In any event, however, $N$ is static and $L$ and $K$ are binary matrices.

The subsystem structure of any system is well defined although it may be trivial (a single internal block) if the system does not decompose naturally into an interconnection of subsystems, such as (\ref{eq:example_l0}) in Figure \ref{fig:linear_example_c} and Figure \ref{fig:SSexample_b}.  Note that $\cal S$ always identifies the most refined subsystem structure possible, and for systems with many interconnected subsystems, coarser representations may be obtained by aggregating various subsystems together and constructing the resulting condensation graph.  These coarser representations effectively absorb some interconnection variables and their associated edges into the aggregated components,  suggesting that such representations are the subsystem structure for less intricate realizations of the system, where some of the manifest auxiliary variables are removed, or at least left hidden.  The subsystem structure is thus a natural partial description of system structure when the system can be decomposed into the interconnection of specific subsystems.

\subsection{Signal Structure}

Another very natural way to partially describe the structure of a system is to characterize the direct causal dependence among each of its manifest variables; we will refer to this notion as the {\em signal structure}.  This description of the structure of a system makes no attempt to cluster, or partition, the actual internal system states.  As a result, it offers no information about the internal interconnection of subsystems, and signal structure can therefore be a very different description of system structure than subsystem structure.  

Given a generalized linear system (\ref{eq:linearsystem}) with complete computational structure $\cal C$, we characterize its signal structure by considering its minimal intricacy realization $(A_o,B_o,C_o,D_o)$.  We assume without loss of generality that the outputs $y$ are ordered in such a way that $C_o$ can be partitioned
\[
C_o = \left[\begin{array}{cc}C_{11}&C_{12}\\C_{21}&C_{22}\end{array}\right]
\]  
where $C_{11}\in \mathbb{R}^{p_1\times p_1}$ is invertible, with  $p_1$  equal to the rank of $C_o$; $C_{12}\in \mathbb{R}^{p_1 \times (n-p_1)}$; $C_{21}\in \mathbb{R}^{(p-p_1)\times p_1}$; and $C_{22} \in \mathbb{R}^{(p-p_1)\times(n-p_1)}$.   Note that if the outputs of the minimal intricacy realization do not result in $C_{11}$ being invertible, then it is possible to reorder them so the first $p_1$ outputs correspond to independent rows of $C_o$; the states can then be renumbered so that $C_{11}$ is invertible.  One can show that such reordering of the outputs and states of the minimal intricacy realization only affects the ordering of the states and outputs of the original system; the graphical relationship of the computational structure is preserved.

The direct causal dependence among manifest variables is then revealed as follows.  First, since $C_o$ is rank $p_1$, the rank-nullity theorem guarantees $C_o$ has nullity $n-{p_1}.$  Let \[ N_{n-{p_1}}=\left[ \begin{array}{c} N_{1} \\ N_2 \end{array} \right], \] where $N_1 \in \mathbb{R}^{p_1 \times (n-{p_1})},\  N_2 \in\mathbb{R}^{(n-{p_1})\times(n-{p_1})},$ be a matrix of column vectors that form a basis for the null space of $C_o.$  By the definition of $N$ and the invertibility of $C_{11}$, we can write \[ N = \left[\begin{array}{c} -C_{11}^{-1}C_{12}N_2 \\ N_2 \end{array}\right].\]  Since the columns of $N$ form a basis, we deduce that $N_2$ is invertible; thus the state transformation $z=Tx$ given by \begin{equation}
T = \left[\begin{array}{cc}C_{11} & C_{12} \\0&N_2^{-1}\end{array}\right].
\end{equation} is well defined. 
This transformation yields a system of the form
\begin{equation}
\label{eq:weaksystem}
\begin{array}{rcl}
\left[\begin{array}{c}\dot{z}_1\\\dot{z}_2\end{array}\right]&=&\left[\begin{array}{cc}A_{11}&A_{12}\\A_{21}&A_{22}\end{array}\right]\left[\begin{array}{c}z_1\\z_2\end{array}\right]+\left[\begin{array}{c}B_1\\B_2\end{array}\right]u\\\\
\left[\begin{array}{c}y_1\\y_2\end{array}\right] &=& \left[\begin{array}{cc}I&0\\C_{2}&0\end{array}\right]\left[\begin{array}{c}z_1\\z_2\end{array}\right]+\left[\begin{array}{c}D_1\\D_2\end{array}\right]u
\end{array}
\end{equation}
where $C_2 = C_{21}C_{11}^{-1},\  z_1\in \mathbb R^{p_1},\ z_2 \in \mathbb{R}^{n-p_1}$, $u\in \mathbb R^m$, $y_1\in \mathbb R^{p_1}$, and $y_2\in \mathbb R^{p-p_1}$.  To simplify the exposition we will abuse notation and refer to the above system as $(A,B,C,D)$, since there is little opportunity to confuse these matrices with those of the original system given in (\ref{eq:linearsystem}).  In fact, the system (\ref{eq:weaksystem}) is simply a change of coordinates of the minimal intricacy realization $(A_o,B_o,C_o,D_o)$, possibly with a reordering of the output and state variables.  The direct causal dependence among manifest variables is then revealed by the dynamical structure function of $(A,B,\left[\begin{array}{cc}I&0\end{array}\right],D_1)$.

The dynamical structure function of a class of linear systems was defined in \cite{ourTrans} and discussed in \cite{yeungcdc,yeungCDC10,howescdc08,ourCDC,yuancdc09,fosbe09}.  This representation of a linear system describes the direct causal dependence among a subset of state variables, and it will extend to characterize signal structure for the system in (\ref{eq:weaksystem}).  We repeat and extend the derivation here to demonstrate its applicability to the system (\ref{eq:weaksystem}).  Taking Laplace transforms and assuming zero initial conditions yields the following relationships
\begin{eqnarray}\label{eq:laplace}
\left[\begin{array}{c}sZ_1\\sZ_2\end{array}\right]&=&\left[\begin{array}{cc}A_{11}&A_{12}\\A_{21}&A_{22}\end{array}\right]\left[\begin{array}{c}Z_1\\Z_2\end{array}\right]+\left[\begin{array}{c}B_1\\B_2\end{array}\right]U
\end{eqnarray}
where $Z(s)$ denotes the Laplace transform of $z(t)$, etc.  Solving for $Z_2$ in the second equation and substituting into the first then yields
\begin{equation}\label{eq:WV}
sZ_1 = W(s)Z_1+V(s)U
\end{equation}  
where $W(s) = \left[A_{11}+A_{12}(sI-A_{22})^{-1}A_{21}\right]$ and $V(s)=\left[B_1+A_{12}(sI-A_{22})^{-1}B_2\right]$.  Let $\hat{D}(s)$ be the matrix of the diagonal entries of $W(s)$, yielding
\begin{equation}
\label{eq:QP}
Z_1 = Q(s)Z_1+P(s)U
\end{equation}
where $Q(s)=(sI-\hat{D}(s))^{-1}(W(s)-\hat{D}(s))$ and $P(s)=(sI-\hat{D}(s))^{-1}V(s)$.  From (\ref{eq:weaksystem}) we note that $Z_1=Y_1-D_1U$, which, substituting into (\ref{eq:QP}), then yields:
\begin{equation}
\label{eq:DSF}
\left[\begin{array}{c}Y_1\\Y_2\end{array}\right] = \left[\begin{array}{c}Q(s)\\C_{2}\end{array}\right]Y_1+\left[\begin{array}{c}P(s)+(I-Q(s))D_1\\D_2\end{array}\right]U
\end{equation}
We refer to the matrices \[ \left[\begin{array}{cc}Q(s)^T & C_2^T \end{array}\right]^T \mbox{   }   \left[\begin{array}{cc}(P(s) + (I-Q(s))D_1)^T & D_2^T \end{array}\right]^T \]  as $\bar{Q}$  and $ \bar{P},$ respectively.  The matrices $(Q(s),P(s))$ are called the dynamical structure function of the system (\ref{eq:weaksystem}), and they characterize the dependency graph among manifest variables as indicated in Equation (\ref{eq:DSF}).   We note a few characteristics of $(Q(s),P(s))$ that give them the interpretation of system structure, namely:
\begin{itemize}
\item $Q(s)$ is a square matrix of strictly proper real rational functions of the Laplace variable, $s$, with zeros on the diagonal.  Thus, if each entry of $y_1$ were the node of a graph, $Q_{ij}(s)$ would represent the weight of a directed edge from node $j$ to node $i$; the fact $Q_{ij}(s)$ is proper preserves the meaning of the directed edge as a {\em causal} dependency of $y_i$ on $y_j$.  
\item Similarly, the entries of the matrix $\left[P(s)+(I-Q(s))D_1\right]$ are proper and thus carry the interpretation of causal weights characterizing the dependency of entries of $y_1$ on the $m$ inputs, $u$.  Note that when $D_1=0$, this matrix reduces to $P(s)$, which has {\em strictly} proper entries.
\end{itemize}
This leads naturally to the definition of signal structure.

\begin{definition}  The {\em signal structure} of a system $G$, with realization (\ref{eq:linearsystem}) and equivalent realization (\ref{eq:weaksystem}), and with dynamical structure function $(Q(s),P(s))$ characterized by (\ref{eq:QP}), is a directed graph $\cal W$, with a vertex set $V({\cal W})$ and edge set $E(\cal W)$ given by:
\begin{itemize}
\item $V({\cal W})=\{u_1,...,u_m,y_{11},...,y_{1p_1},y_{21},...,y_{2p_2}\}$, each representing a manifest signal of the system, and
\item $E({\cal W})$ has an edge from $u_i$ to $y_{1j}$, $u_i$ to $y_{2j}$, $y_{1i}$ to $y_{1j}$ or $y_{1i}$ to $y_{2_j}$ if the associated entry in $\left[P(s)+(I-Q(s))D_1\right]$, $D_2$, $Q(s)$, or $C_{21}$ (as given in Equations (\ref{eq:QP}) and (\ref{eq:DSF})) is nonzero, respectively. 
\end{itemize}
We label the nodes of $V({\cal W})$ with the name of the associated variable, and the edges of $E({\cal W})$ with the associated transfer function entry from Equation (\ref{eq:DSF}).
\end{definition}

Signal structure is fundamentally a different {\em type} of graph than either the computational or subsystem structure of a system because, unlike these other graphs, vertices of a system's signal structure represent {\em signals} rather than systems.  Likewise, the edges of $\cal W$ represent {\em systems} instead of signals, as opposed to $\cal C$ or $\cal S$.  We highlight these differences by using circular nodes in $\cal W$, in contrast to the square nodes in $\cal C$ or $\cal S$.  The next example illustrates a system without notable subsystem structure and no apparent structural motif in its complete computational structure; nevertheless, it reveals a simple and elegant ring structure in its signal structure.  

\begin{example} {\em Ring Signal Structure.}
\label{exampleweakring}
Systems with no apparent structure in any other sense can nevertheless possess a very particular signal structure. Consider the minimally intricate linear system, specified by the state-space realization $(A_o,B_o,C_o,D_o)$, where 

\begin{equation*}
\label{eq:weaka}
A_o  = \renewcommand{\tabcolsep}{.5mm}\frac{1}{12}\left[\begin{array}{cccccc} -178 & 262 & -10 &  -141& -19  & 88 \\[.05in] -156 &  252& -12& -156 &-48  & 84 \\[.05in] -158 &  266&  -38& -147&  -5& 128 \\[.05in]  -12& 48& -12 & -72 &-12  &12\\[.05in]  -288&  504& 0& -264 & -180 & 144 \\[.05in] 0 &24  &0  &-24  &-12  &-12  \end{array}\right], 
\end{equation*}
\begin{equation*}
\label{eq:weakb}
B_o = \frac{1}{4}\left[\begin{array}{ccc} 0 & -1 & 21 \\ 0 & 0 & 12 \\ -8 & 1 & 27 \\ 0 & 0 & 0 \\ 0 & 0 & 0 \\ 0 & 0 & 0\end{array} \right],
\end{equation*}
\begin{equation}
\label{eq:weakc}
C_o = \begin{bmatrix} -1 & 4 & -1 & -2 & -1 & 1 \\[.05in] -12 & 21 & 0 & -11 & -5 & 6 \\[.05in] 0 & 2& 0 & -2 & -1 & 0 \end{bmatrix}\hspace{-.1cm},\; D_o = \left[\begin{array}{ccc}0&0&0\\0&0&0\\0&0&0\end{array}\right]\hspace{-.1cm}.
\end{equation}

\begin{figure}[htb]
\centering
  \includegraphics*[width=.45\textwidth, viewport = 100 50 720 550]{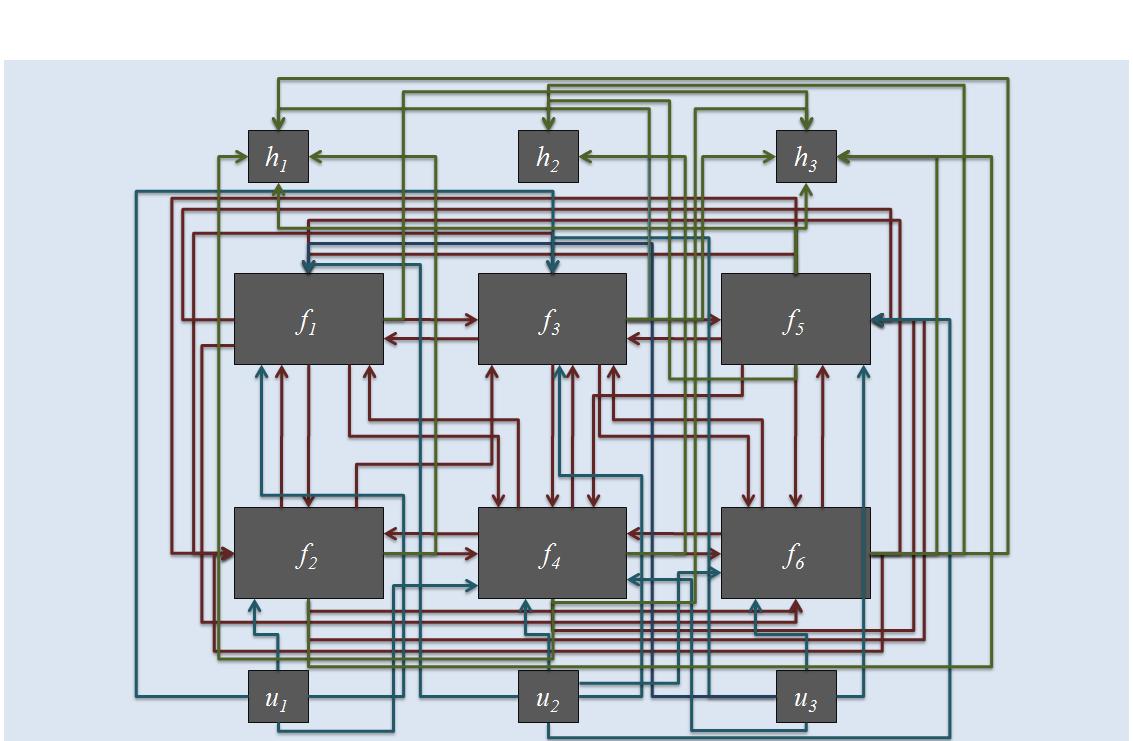}
\caption{The complete computational structure ${\cal C}$,  of the system $(A_o, B_o, C_o,D_o)$ given by (\ref{eq:weakc}).  Here, edge labels, $u$ and $x$, and self loops on each node $f_i$ have been omitted to avoid the resulting visual complexity.  Edges associated with variables $x$, which are not manifest, are entirely contained within the shaded region (which also corresponds to the strong partial condensation shown in Figure \ref{fig:cyclicweakstrong}).} 
 \label{fig:cyclicweakcomp}     
\end{figure}

\begin{figure}[htb]
\centering
     \subfigure[The subsystem structure, ${\cal S}$, of the system shown in Figures \ref{fig:cyclicweakcomp} and \ref{fig:cyclicweak}; the system has no interconnection structure between subsystems because the system is composed of only a single subsystem; it does not meaningfully decompose into smaller subsystems. ]{
          \includegraphics*[width=.45\textwidth, viewport = 120 190 460 420]{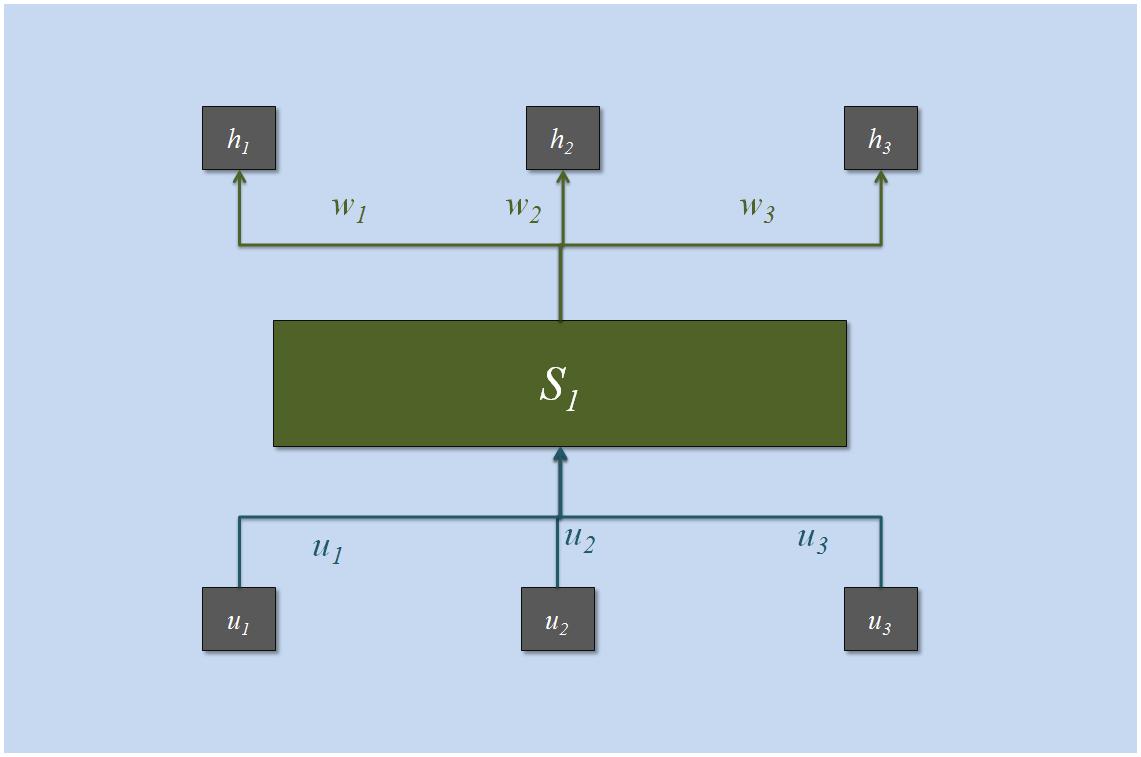}
          \label{fig:cyclicweakstrong}
          }\\
          \vspace{-.5cm}
           \subfigure[The signal structure, ${\cal W}$, of the system shown in Figures \ref{fig:cyclicweakcomp} and \ref{fig:cyclicweakstrong}.  Note that, in contrast with $\cal C$ and $\cal S$, vertices of $\cal W$ represent manifest {\em signals} (characterized by round nodes), while edges represent {\em systems}.]{
          \includegraphics*[width=.45\textwidth, viewport = 160 130 635 600]{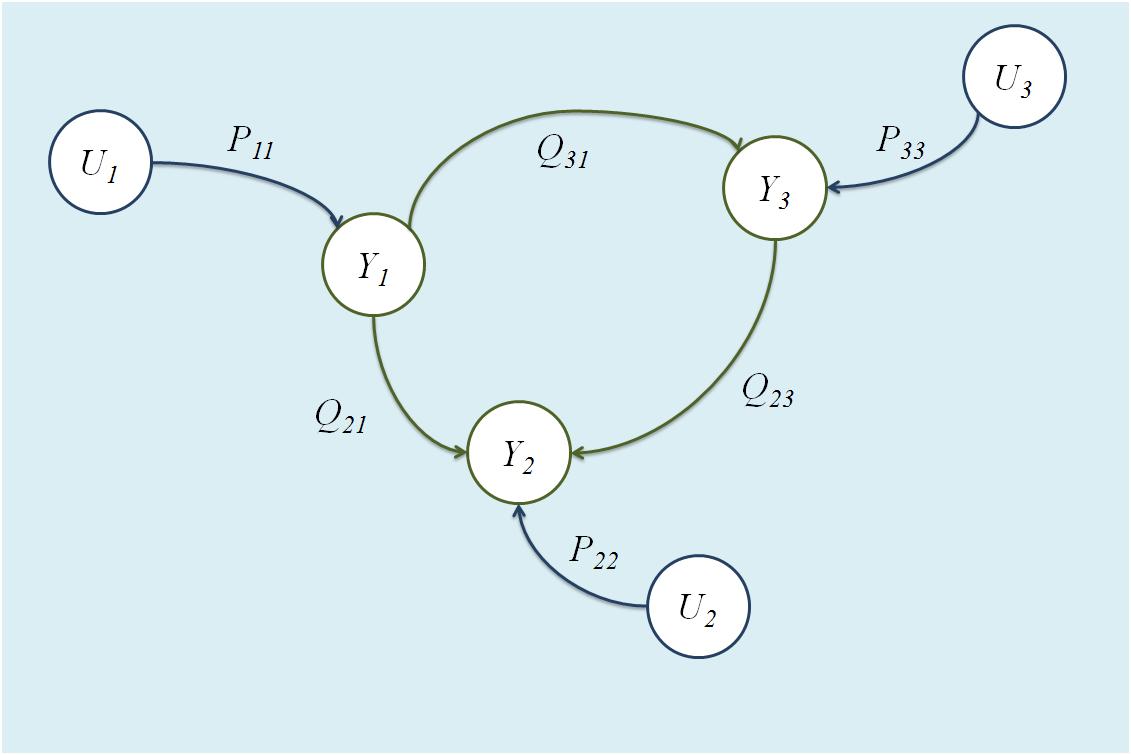}
          \label{fig:cyclicweak}
          }
\caption{Distinct notions of partial structure for the system $(A_o, B_o, C_o, D_o)$ given by Equation (\ref{eq:weakc}).  Observe that while the system is unstructured in the strong sense, it is actually quite structured in the weak sense.  Moreover, although the subsystem structure is visible in the complete computational structure (Figure \ref{fig:cyclicweakcomp}) as a natural condensation (note the shaded region), the signal structure is not readily apparent. }
\label{fig:cyclicweakorig}
\end{figure}

\noindent We compute the signal structure by employing a change of coordinates on the state variables to find an equivalent realization of the form (\ref{eq:weaksystem}).  The transformation
 \begin{equation}
 T = \left[ \begin {array}{cccccc}  -1 & 4 & -1 & -2 & -1 & 1 \\\noalign{\smallskip} -12 & 21 & 0 & -11 & -5 & 6 \\\noalign{\smallskip} 0 & 2 & 0 & -2 & -1 & 0 \\\noalign{\smallskip} 0 & 0 & 0 & 1 & 0 & 0 \\\noalign{\smallskip} 0 & 0  & 0 & 0 & 1 & 0 \\\noalign{\smallskip} 0 & 0 & 0 & 0 & 0 & 1 \end {array} \right]
 \end{equation}
 results in the following realization
 \begin{equation*}
 A=TA_0T^{-1} = \left[ \begin {array}{cccccc} - 2 & 0 & 0 & 0 & 0 &- 3 \\\noalign{\smallskip} 0 &- 3 & 0 &- 1 & 0 & 0 \\\noalign{\smallskip} 0 & 0 &- 4 & 0 & 5 & 0 \\\noalign{\smallskip} 1 & 0 & 0 &- 4 & 0 & 0 \\\noalign{\smallskip} 0 & 2 & 0 & 0 &- 5 & 0 \\\noalign{\smallskip} 0 & 0 & 1 & 0 & 0 &- 1 \end {array} \right]
 \end{equation*}
 \begin{equation}
 \begin{array}{lr} 
B=TB_0 = \left[ \begin {array}{ccc}  2 & 0 & 0 \\\noalign{\smallskip} 0 & 3 & 0 \\\noalign{\smallskip} 0 & 0 & 6 \\\noalign{\smallskip} 0 & 0 & 0 \\\noalign{\smallskip} 0 & 0 & 0 \\\noalign{\smallskip} 0 & 0 & 0 \end {array} \right],&
\begin{array}{c}
C=C_0T^{-1} = \left[\begin{array}{cc} I_3 & 0 \end{array}\right], \\\\\\
D = D_0=[0].
\end{array}
\end{array}
 \end{equation}
The dynamical structure function of the system, $(Q,P)$, then becomes 
\begin{equation*} Q=
 \left[ \begin{array}{ccc} 0&0&\frac{-3}{{s}^{2}+3s+2}\\\noalign{\smallskip} \frac{-1} {{s}^{2}+7s+12}&0&0\\\noalign{\smallskip}0&\frac{10}
  {{s}^{2}+9s+20}&0\end{array} \right],
  \end{equation*}
  \begin{equation}
   P =  \left[ \begin {array}{ccc} \frac{2}{s+2}&0&0\\\noalign{\smallskip} 0&\frac{3}{s+3}&0\\\noalign{\smallskip}0&0&\frac{6}{s+4}\end{array} \right]
  \end{equation}
which yields the signal structure, $\cal W$, as shown in Figure \ref{fig:cyclicweak}.  Notice that although the complete computational structure and subsystem structure do not characterize any meaningful interconnection patterns, the system is nevertheless structured and organized in a very concrete sense.  In particular, the outputs $y_1$, $y_2$, and $y_3$ form a cyclic dependency, and each causally depends on a single input $u_i, i = 1,2,3,$, respectively. 
\end{example}

\subsection{Sparsity Structure of the Transfer Function Matrix}

The weakest notion of structure exhibited by a system is the pattern of zeros portrayed in its transfer function matrix, where ``zero" refers to the value of the particular transfer function element, not a transmission zero of the system.  Like signal structure, this type of structure is particularly meaningful for multiple-input multiple-output systems, and, like signal structure, the corresponding graphical representation reflects the dependance of system output variables on system input variables.  Thus, nodes of the graph will be signals, represented by circular nodes, and the edges of the graph will represent systems, labeled with the corresponding transfer function element; a zero element thus corresponds to the absence of an edge between the associated system input and output.  Formally, we have the following definition.

\begin{definition}  The {\em sparsity structure} of a system $G$ is a directed graph $\cal Z$, with a vertex set $V({\cal Z})$ and edge set $E(\cal Z)$ given by:
\begin{itemize}
\item $V({\cal Z})=\{u_1,...,u_m,y_{1},...,y_{p}\}$, each representing a manifest signal of the system, and
\item $E({\cal Z})$ has an edge from $u_i$ to $y_{j}$ if $G_{ji}$ is nonzero. 
\end{itemize}
We label the nodes of $V({\cal Z})$ with the name of the associated variable, and the edges of $E({\cal Z})$ with the associated element from the transfer function $G(s)$.
\end{definition}

Unlike signal structure, note that the sparsity structure of the transfer function matrix describes the closed-loop dependency of an output variable on a particular input variable, not its {\em direct} dependence.  As a result, the graph is necessarily bipartite, and all edges will begin at an input node and terminate at an output node; no edges will illustrate dependencies between output variables.   For example, the sparsity structure for the system in Example \ref{exampleweakring}, is shown in Figure \ref{fig:zpweakring}, with transfer function $G(s)=$ 

{\tiny 
\begin{equation}
\label{eq:TFweakring}
\left[ \begin {array}{ccc}{\frac {n_1(s)}{d(s) }}&  {\frac {- 90(s+ 4)}{d(s) }}& {\frac {- 18 (s^{3}+ 12  s^{2}+ 47  s+ 60 )}{d(s) }}\\\noalign{\medskip}  {\frac {- 2(s^{3}+ 10  s^{2}+ 29  s+ 20) }{d(s) }}&  {\frac {3 (s^{5}+ 16  s^{4}+ 97  s^{3}+ 274  s^{2}+ 352  s+ 160) }{d(s) }}&  {\frac {18(s+ 5) }{d(s) }}\\\noalign{\medskip} {\frac {- 20(s+ 1 )}{d(s) }}&   {\frac {30(s^{3}+ 7  s^{2}+ 14  s+ 8) }{d(s) }}&   {\frac {n_2(s) }{d(s) }}\end {array} \right] ,
\end{equation}}

\noindent where $d(s) = s^{6}+ 19  s^{5}+ 145  s^{4}+ 565  s^{3}+ 1174  s^{2}+ 1216  s+ 450,\   n_1(s)=2(s^{5}+ 17  s^{4}+ 111  s^{3}+ 343  s^{2}+ 488  s+ 240 )$, and $ n_2(s) = 6(s^{5}+ 15  s^{4}+ 85  s^{3}+ 225  s^{2}+ 274  s+ 120).$  
\begin{figure}[htb]
\centering
  \includegraphics*[width=.45\textwidth]{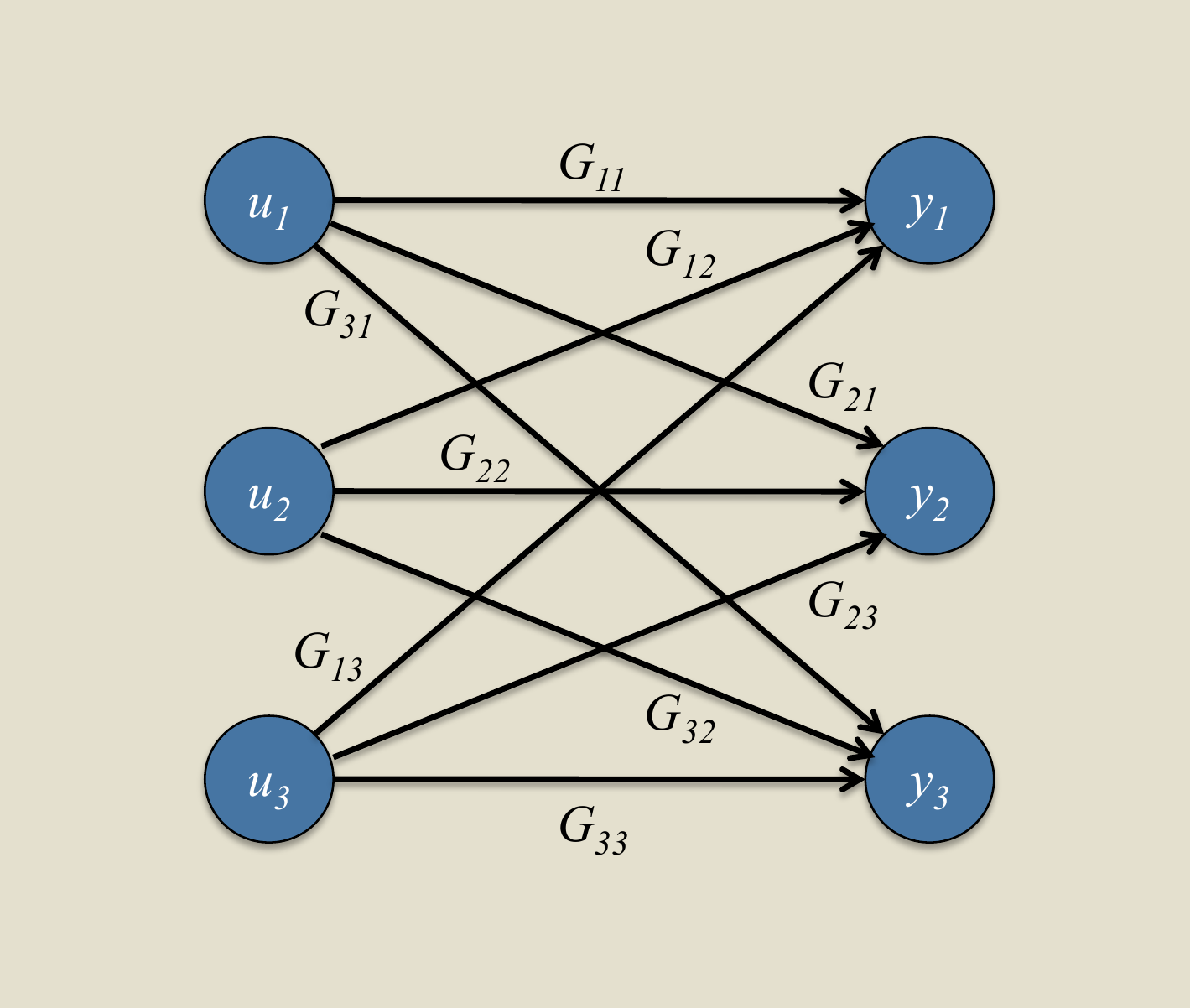}
\caption{The sparsity structure for the system in Example \ref{exampleweakring}.  Note that vertices of sparsity structure are manifest {\em signals}, distinguished in this work by circular nodes similar to those of signal structure.  Nevertheless, this system has a fully connected (and thus unstructured) sparsity structure, while its signal structure (shown in Figure \ref{fig:cyclicweak}) exhibits a definite ring structure. } 
 \label{fig:zpweakring}     
\end{figure}
Although the {\em direct} dependence, given by the signal structure, is cyclic (Figure \ref{fig:cyclicweak}), there is a path from every input to every output that does not cancel.  Thus, the sparsity structure is fully connected, corresponding to the full transfer function matrix in Equation (\ref{eq:TFweakring}).

It is important to understand that the sparsity structure does not necessarily describe the flow of information between inputs and outputs.   The presence of a zero in the $(i,j)^{\rm th}$ location simply indicates that the input-output response of the system results in the $i^{\rm th}$ output having no dependence on the $j^{\rm th}$ input.  Such an effect, however, could be the result of certain internal cancellations and does not suggest, for example, that there is no path in the complete computational structure from the $j^{\rm th}$ input to the $i^{\rm th}$ output.  Thus, for example, a diagonal transfer function matrix does not imply the system is decoupled.     

The next example demonstrates that a fully coupled system may, nevertheless, have a diagonal transfer function, even when the system is minimal in both intricacy and order.  That is, the internal cancellations necessary to generate the diagonal sparsity structure in the transfer function while being fully coupled do not result in a loss of controllability or observability.   Thus, the sparsity structure only describes the input-output dependencies of the system and does not imply anything about the internal information flow or communication structure.   This fact is especially important in the context of decentralized control, where the focus is often on shaping the sparsity of a controller's transfer function given a particular sparsity in the transfer function of the system to be controlled.

\begin{example}{\em Diagonal Transfer Function $\neq$ Decoupled.}
Consider a system, $G$, with the following minimal intricacy realization, $(A_o,B_o,C_o,D_o)$:
\begin{equation}
\begin{array}{rcl}
\left[\begin{array}{c}\dot{x}_1\\\dot{x}_2\end{array}\right]&=&\left[\begin{array}{cc}-5&1\\2&-4\end{array}\right]\left[\begin{array}{c}x_1\\x_2\end{array}\right]+
\left[\begin{array}{cc}2&1\\4&-1\end{array}\right]\left[\begin{array}{c}u_1\\u_2\end{array}\right]\\\\
\left[\begin{array}{c}y_1\\y_2\end{array}\right]&=&\left[\begin{array}{cc}1&1\\-4&2\end{array}\right]\left[\begin{array}{c}x_1\\x_2\end{array}\right]
\end{array}
\end{equation}
It is easy to see from $(A_o,B_o,C_o,D_o)$ that this system has a fully connected complete computational structure.  Moreover, one can easily check that the realization is controllable and observable, and thus of minimal order.  Nevertheless, its transfer function is diagonal, given by
\begin{equation}
G(s) = \left[\begin{array}{cc}\frac{6}{s+3}&0\\0&\frac{-6}{s+6}\end{array}\right].
\end{equation}
\end{example}
\section{Relationships Among Representations of Structure}
In this section, we explore the relationships between the four representations of structure defined above.  What we will find is that some representations of structure are more informative than others.  Next, we will discuss a specific class of systems in which the four representations of structure can be ranked by information content (with one representation encapsulating all the information contained in another representation of structure).  Outside this class of systems, however, we will demonstrate that signal structure and subsystem structure have fundamental differences in addition to those arising from the graphical conventions of circular versus square nodes, etc.  Signal and subsystem structure provide alternative ways of discussing a system's structure without requiring the full detail of a state-space realization or the abstraction imposed by a transfer function.  Understanding the relationships between these representations then enables the study of new kinds of research problems that deal with the realization, reconstruction, and approximation of system structure (where structure can refer to the four representations defined so far or other representations of structure not mentioned in this work). The final section gives a discussion of such future research. 

Different representations of structure contain different kinds of structural information about the system.  For example, the complete computational structure details the structural dependencies among fundamental units of computation.  Using complete computational structure to model system structure requires knowledge of the parameters associated with each unit of computation. Partial representations of structure such as signal and subsystem structure do not require knowledge of such details in their description. Specifically, subsystem structure essentially aggregates units of computation to form subsystems and models the closed-loop transfer function of each subsystem.  Signal structure models the SISO transfer functions describing direct causal dependencies between outputs and inputs of some of the fundamental units of computation that happen to be manifest.  Sparsity structure models the closed-loop dependencies of system outputs on inputs. Thus, complete computational structure appears to be the most demanding or information-rich description of system structure.  This intuition is made precise with the following result:

\begin{theorem}\label{CCSsaysall}
Suppose a complete computational structure has minimal intricacy realization $(A_o,B_o,C_o,D_o)$ with \[C_0 = \begin{bmatrix} C_{11}& C_{12} \\ C_{21} & C_{22} \end{bmatrix}\] and $C_{11}$ invertible. 
Then the complete computational structure specifies a unique subsystem, signal, and sparsity structure.
\end{theorem}
\begin{proof}
Let ${\cal C}$ be a computational structure with minimal intricacy realization $(A_o,B_o,C_o,D_o)$ with  \[ C_0 = \begin{bmatrix} C_{11}& C_{12} \\ C_{21} & C_{22} \end{bmatrix},\] $C_{11}$ invertible.  By Lemma \ref{spsunique}, the subsystem structure ${\cal S}$ is unique.   Since $C_{11}$ is invertible, we see that equations (\ref{eq:WV}) and (\ref{eq:DSF}) imply that the minimal intricacy realization uniquely specifies the dynamical structure function of the system.  By definition, the signal structure is unique.  Finally, write the generalized state-space realization of ${\cal C}$ as 
\[(\bf{A},\bf{B},\bf{C},\bf{D}) =  \left(\left[\begin{array}{cc} A & \hat{A} \\ \bar{A} & \tilde{A}\end{array}\right],	\left[\begin{array}{c}B \\ \bar{B} \end{array}\right] , \left[\begin{array}{cc} C & \bar{C} \end{array}\right],D \right).\]  The uniqueness of the sparsity  structure follows from its one-to-one correspondence with the transfer function $G(s) = C_o(sI-A_o)^{-1}B_o+D_o$ which can also be expressed as 
\[ (C+ \bar{C}(I-\tilde{A})^{-1}\bar{A})(sI-A - \hat{A}(I-\tilde{A})^{-1}\bar{A})^{-1}(B+\bar{B}(I-\tilde{A})^{-1}\bar{A}).\] 
\end{proof}

It is well known that a transfer function $G(s)$ can be realized using an infinite number of state-space realizations. Without additional assumptions, e.g. full state feedback, it is impossible to uniquely associate a single state-space realization with a given transfer function.  On the other hand, a state space realization specifies a unique transfer function.  In this sense, a transfer function contains less information than the state space realization.   

Similarly,  subsystem, signal, and sparsity structure can be realized using multiple complete computational structures.    Without additional assumptions, it is impossible to associate a unique complete computational structure with a given subsystem, signal, or sparsity structure.  Theorem \ref{CCSsaysall} shows that a complete computational structure specifies a unique subsystem, signal, and sparsity structure. In this sense, a complete computational structure is a more informative description of system structure than subsystem, signal and sparsity structure.   The next result has a similar flavor and follows directly from the one-to-one correspondence of a system's transfer function with its sparsity structure.
\begin{theorem}\label{SQPtoG}
Every subsystem structure or signal structure specifies a unique sparsity structure.
\end{theorem}
\begin{proof}
Consider the LFT representation ${\cal F}(N,S)$ of a subsystem structure ${\cal S}$; write \[ N = \left[\begin{array}{c|c} 0 & I \\\hline L & K \end{array}\right]\] as in equation (\ref{eq:LFT}).  The linear fractional transform gives the input-output map, i.e. the transfer function.  Thus, $G(s) = (I-SK)^{-1}SL.$ 

Similarly, using the dynamical structure representation of the signal structure ${\cal W}$ given in equation (\ref{eq:DSF}), we can solve for the transfer function 
\[ G(s) = \left(I-\left[\begin{array}{cc}Q(s)& 0 \\C_{2} & 0\end{array}\right]\right)^{-1}\left[\begin{array}{c}P(s)+(I-Q(s))D_1\\D_2\end{array}\right] \]    The result follows from the definition of sparsity structure.
\end{proof}

The relationship between subsystem structure and signal structure is not so straightforward.  Nevertheless, subsystem structure does specify a unique signal structure for a class of systems, namely systems with subsystem structure composed of single output (SO) subsystems and where every manifest variable is involved in subsystem interconnection.  For this class of systems, subsystem structure is a richer description of system structure than signal structure. 

{\subsection{Single Output Subsystem Structure and Signal Structure}
\noindent \begin{theorem}\label{SpecialMISO}
Let ${\cal S}$ be a  SO subsystem structure with LFT representation ${\cal F}(N,S).$  Suppose \[N = \left[\begin{array}{c|c}0 & I \\\hline L & K \end{array} \right],\]  where $\begin{bmatrix} L & K \end{bmatrix}$ has full column rank. Then ${\cal S}$ uniquely specifies the signal structure of the system.
\end{theorem}
\begin{proof}
The dynamics of $N$ and $S$ satisfy \begin{eqnarray}\label{eq:NSdynamics} Y  &=& \left[0\right] U + \left[I\right] Y \\ \pi &=& L U + K Y \\ Y &=& S\pi \end{eqnarray}  Combining the second and third equation, we get \[Y = S\pi = S\begin{bmatrix}K & L \end{bmatrix}\begin{bmatrix} Y \\ U \end{bmatrix}.\]  Since ${\cal S}$ is a SO subsystem structure, the entries of $S$ describe direct causal dependencies among manifest variables involved in interconnection. Since $\begin{bmatrix} L & K \end{bmatrix}$ has full column rank and is a binary matrix, this means that each manifest variable is used at least once in interconnection.   Thus, $S$ describes the direct causal dependencies between all manifest variables of the system and specifies the signal structure of the system.
\end{proof}
Notice that for the class of systems described above, the four representations of structure can be ordered in terms of information content.  Theorem \ref{CCSsaysall} shows that the complete computational structure uniquely specifies all the other representations of structure and thus is the most informative of the four.  By Theorem \ref{SpecialMISO} and Theorem \ref{SQPtoG} respectively, subsystem structure uniquely specifies the signal structure and sparsity structure of the system and thus is the second most informative.  Similarly, signal structure is the third most informative and sparsity structure is the least informative of the four representations of structure.  

We note that the converse of Theorem \label{theorem:SpecialMISO} is also true, namely if the subsystem structure of a system specifies a unique signal structure then the subsystem structure is a SO subsystem structure where every manifest variable is an interconnection variable.  The proof is simple and follows from the result of the next subsection.  We also provide several examples that show how a multiple output subsystem structure can be consistent with multiple signal structures - these all will serve to illustrate the general relationship between subsystem and signal structure outside of the class of systems mentioned above.}

\subsection{The Relationship Between Subsystem and Signal Structure}
Subsystem structure and signal structure are fundamentally different descriptions of system structure.  In general, subsystem structure does not encapsulate the information contained in signal structure.  Signal structure describes direct causal dependencies between manifest variables of the system.  Subsystem structure describes closed loop dependencies between manifest variables involved in the interconnection of subsystems.  Both representations reveal different perspectives of a system's structure. The next result makes this relationship between subsystem and signal structure precise.\\ 
\begin{theorem}
Given a system $G$, let ${\cal F}(N,S)$ be the LFT representation of a subsystem structure ${\cal S}.$ In addition, let the signal structure of the system $G$ be denoted as in equation (\ref{eq:DSF}).  Let $Y(S_i)$ denote the outputs associated with subsystem $S_i.$ Define \[\left[ Q_{int}(s)\right]_{ij} \equiv \begin{cases} \bar{Q}_{ij}(s) & y_i, y_j \in Y(S_k), S_k \text{ a subsystem in ${\cal S}$}  \\ 0 & \text{ otherwise,} \end{cases} \] and $Q_{ext} \equiv \bar{Q}(s) - Q_{int}(s).$ Then the signal structure and subsystem structure are related in the following way:
\begin{equation}\label{eq:SQPrelation} S \left[ \begin{array}{c|c} L & K \end{array}\right] = (I-Q_{int})^{-1} \left[ \begin{array}{c|c}  \bar{P}& Q_{ext} \end{array} \right] \end{equation}
\end{theorem}
\begin{proof}
Examining relation (\ref{eq:SQPrelation}), observe that the $ij^{\text{th}}$ entry of the left hand side describes the closed loop causal dependency from the $j^{\text{th}}$ entry of $\left[\begin{array}{cc} U^T & Y^T \end{array}\right]^T$ to $Y_i.$ By closed loop, we mean that they do not describe the internal dynamics of each subsystem, e.g. the direct causal dependencies among outputs of a single subsystem. Thus, these closed loop causal dependencies are obtained by solving out the intermediate direct causal relationships, i.e. the entries in $Q_{int}.$  Notice that the right hand side of (\ref{eq:SQPrelation}) also describes the closed loop map from $\begin{bmatrix} U^T & Y^T \end{bmatrix}^T$ to $Y,$  and in particular the $ij^{\text{th}}$ entry of \[ (I-Q_{int})^{-1} \left[ \begin{array}{c|c}  \bar{P}& Q_{ext} \end{array} \right] \]  describes the closed loop causal dependency from the $j$th entry of $\left[\begin{array}{cc} U & Y \end{array}\right]^T$ to $Y_i.$
\end{proof}

As a special case, notice that for SO subsystem structures, $Q_{int}$ becomes the zero matrix and that for subsystem structures with a single subsystem, $S$ becomes the system transfer function, $L$ becomes the identity matrix, $Q_{int} = \bar{Q}$, and $Q_{ext}$ and $K$ are both zero matrices.  The primary import of this result is that a single subsystem structure can be consistent with two or more signal structures and that a single signal structure can be consistent with two or more subsystem structures.  Consider the following examples:
\noindent\begin{example}{\em A Signal Structure consistent with two Subsystem Structures}\\
In this example, we will show how a signal structure can be consistent with two subsystem structures.   To do this we construct two different generalized state-space realizations that yield the same minimal intricacy realization but different admissible partitions (see Definition \ref{admissiblepartition}). The result is two different subsystem structures that are associated with the same signal structure.  First, we consider the complete computational structure ${\cal C}_1$ with generalized state-space realization \[({\bf{A_1},\bf{B_1},\bf{C_1},\bf{D_1}}) =  \left(\left[\begin{array}{cc} A_1 & \hat{A}_1 \\ \bar{A}_1 & \tilde{A}_1\end{array}\right],	\left[\begin{array}{c}B_1 \\ \bar{B}_1 \end{array}\right] , \left[\begin{array}{cc} C_1 & \bar{C}_1 \end{array}\right],D_1 \right) \] where \[ \begin{array}{cc}
 A_1 = \begin{bmatrix}-4 & 1 & 0 & 0 & 1 \\ 1 & -7 & 0 & 0 & 3 \\ 0 & 0 & -6 & 0 & 0 \\ 0 & 0 & 0 & -6 & 0 \\ 1 & 2 & 0 & 0 & -10 \end{bmatrix}, &  
 
\hat{A}_1 = \begin{bmatrix} 0 & 0 & 2 & 1 \\ 0  & 0 & 2 & 1 \\ 2 & 1 & 0 & 1 \\ 1 & 2 &2 & 0 \\ 0 & 0 & 0 & 0 \end{bmatrix},  \\\noalign{\medskip}  
\bar{A}_1 = \begin{bmatrix} 1 & 0 & 0 & 0 & 0 \\ 0 & 1 & 0 & 0 & 0 \\ 0 & 0 & 1 & 0 & 0 \\ 0 & 0 & 0 & 1 & 0 \\ 0 & 0 & 0 & 0 & 0 \end{bmatrix}, & 
 
\tilde{A}_1 = \left[ 0 \right]_{4},\\\noalign{\medskip} 
B_1 = \left[\begin{array}{ccccc}1 & 1 & 1 & 1 &1 \end{array}\right]^T, &   \bar{B}_1 = \left[\begin{array}{cccc}0&0&0&0\end{array}\right],\\\noalign{\medskip} 
C_1 = \bf{0}_{4\times 5} ,&  \bar{C}_1 = I_4, 
\end{array}
 \] and $D_1 = \bf{0}_{4\times 1}.$  Figure \ref{fig:c1}  shows the computational structure ${\cal C}_1.$  
\begin{figure}[htb!]
\centering
\subfigure[The complete computational structure ${\cal C}_1$ with generalized state-space realization $({\bf{A_1},\bf{B_1},\bf{C_1},\bf{D_1}}).$ We have omitted the self-loops on each of the $f$ vertices for the sake of visual clarity.]{
    \includegraphics[width=\columnwidth]{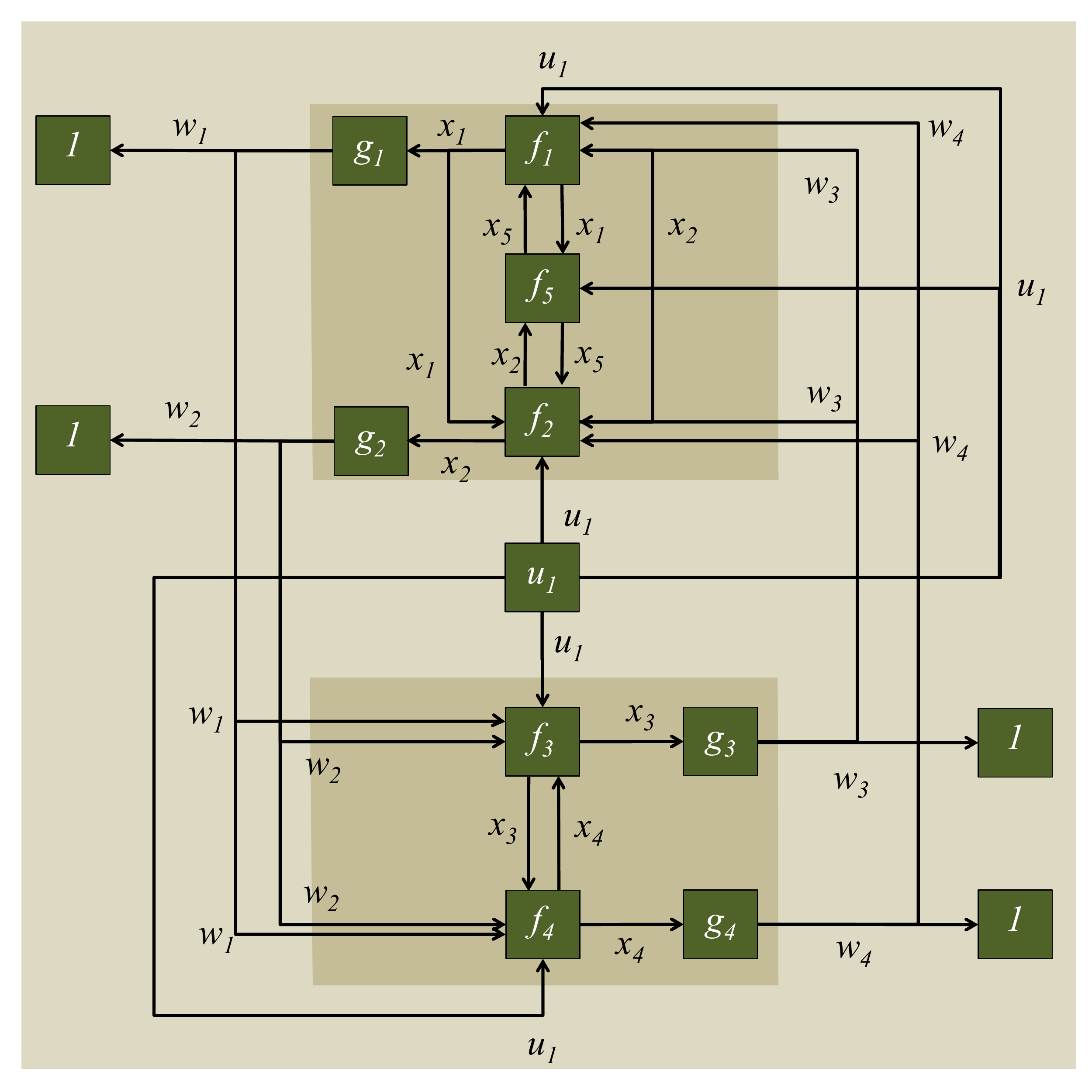}
    \label{fig:c1}}
\subfigure[The complete computational structure ${\cal C}_2$ with generalized state-space realization $({\bf{A_2},\bf{B_2},\bf{C_2},\bf{D_2}}).$ We have omitted the self-loops on each of the $f$ vertices for the sake of visual clarity.]{
    \includegraphics[width=\columnwidth]{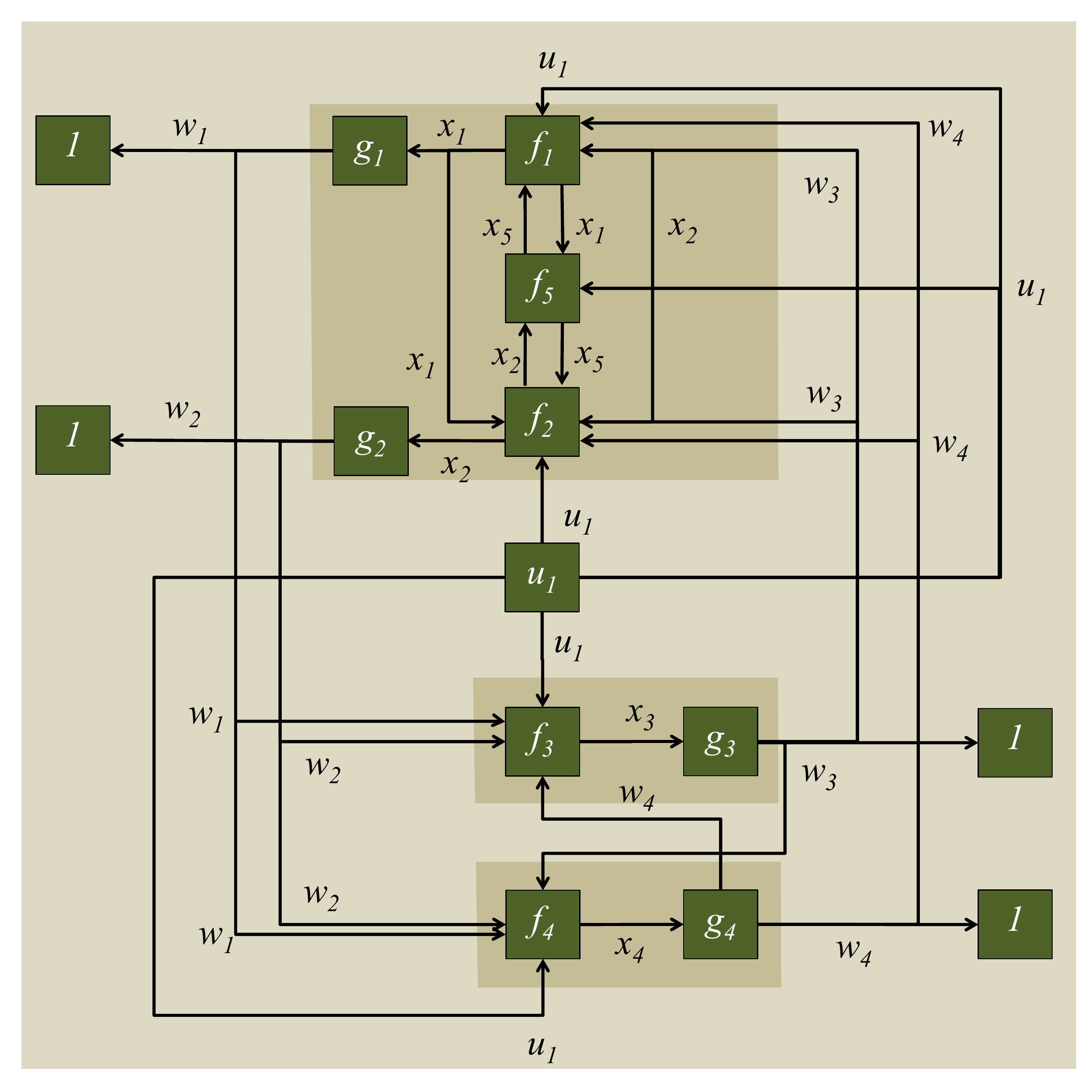}
    \label{fig:c2}}
    \caption{The complete computational structures of two systems that differ only by a slight rearrangement of their state-space dynamics.     The rearrangement results in two different subsystem structure representations.   Nonetheless, the resulting minimal intricacy realization for both systems is the same, implying that both systems have an identical signal structure.  }
\end{figure}

Next consider the complete computational structure ${\cal C}_2$ with generalized state-space realization \[({\bf{A_2},\bf{B_2},\bf{C_2},\bf{D_2}}) =  \left(\left[\begin{array}{cc} A_2 & \hat{A}_2 \\ \bar{A}_2& \tilde{A}_2\end{array}\right],	\left[\begin{array}{c}B_2 \\ \bar{B}_2 \end{array}\right] , \left[\begin{array}{cc} C_2 & \bar{C}_2 \end{array}\right],D_2 \right). \] The matrix entries of this generalized state-space realization are given as follows: 
\[ \begin{array}{ccl}
 A_2 = \begin{bmatrix}-4 & 1 & 0 & 0 & 1 \\ 1 & -7 & 0 & 0 & 3 \\ 0 & 0 & -6 & 1 & 0 \\ 0 & 0 & 2 & -6 & 0 \\ 1 & 2 & 0 & 0 & -10 \end{bmatrix}, &  
 & 
\hat{A}_2 = \begin{bmatrix} 0 & 0 & 2 & 1 \\ 0  & 0 & 2 & 1 \\ 2 & 1 & 0 & 0 \\ 1 & 2 & 0 & 0 \\ 0 & 0 & 0 & 0 \end{bmatrix}  \\ \end{array} \]
\[ \begin{array}{ccl}\bar{A}_2 =\bar{A}_1, & 
 & 
\tilde{A}_2 = \left[ 0 \right]_{4},\\\noalign{\medskip} 
B_2 = B_1 = \bf{1}_{5\times 1}, & &  \bar{B}_2 = \bar{B}_1 = \bf{0}_{4\times 1},\\\noalign{\medskip} 
C_2 = C_1, & & \bar{C}_2 = \bar{C}_1, 
\end{array}
 \] and $D_2 = D_1.$  Figure \ref{fig:c2}  shows the computational structure ${\cal C}_2.$  The difference between these two computational structures is evident more in the subsystem structure representation of the system - note how replacing $A_1$ with $A_2$, essentially externalizes internal dynamics.  The result is that ${\cal C}_2$ admits a subsystem structure ${\cal S}_2$ which  divides one of the subsystems of ${\cal S}_1$ into two subsystems.  This is more apparent in the LFT representations of ${\cal S}_1$ and ${\cal S}_2;$ the LFT representation of ${\cal S}_1$ is given by ${\cal F}(N_1,S_1) = $ \[ N_1 = \left[\begin{array}{c|c} \bf{0}_{4\times 1}  & \bf{I}_{4} \\\hline \begin{array}{c}0\\0\\1\\0\\ 0\\0\\ 1\\ 0\\ 0\\ 0\\ 1 \end{array} & 
\begin{array}{cccc}0 & 0  & 1 & 0 \\ 0 & 0 &  0 & 1  \\ 0 & 0 & 0 & 0 \\ 1 & 0 & 0 & 0 \\ 0 & 1 & 0 & 0 \\ 0 & 0 & 0 & 1 \\ 0 & 0 & 0 & 0 \\ 1 & 0 & 0 & 0 \\ 0 & 1 & 0 & 0 \\ 0& 0 & 1 & 0 \\ 0 & 0 & 0 & 0 \end{array} 
\end{array}\right] \] 
and \[\begin{array}{c}  S_1(s)  = \begin{bmatrix} S_{11}(s) & 0 & 0 \\ 0 & S_{12}(s) & 0 \\ 0 & 0 & S_{13}(s) \end{bmatrix},\end{array} \] with $S_{11}(s), S_{12}(s),S_{13}(s)$ given by \[ \begin{array}{c} \\\noalign{\medskip}
\begin{bmatrix} \frac{2(s^2+18s+76)}{s^3+21s+130s+234} & \frac{s^2+18s+76}{s^3+21s+130s+234} & \frac{s^2+19s+86}{s^3+21s^2+130s+234} \\ \frac{2(s^2+15s+52)}{s^3+21s^2+130s+234} & \frac{s^2+15s+52}{s^3+21s^2+130s+234} & \frac{(13+s)(s+5)}{s^3+21s^2+130s+234} \end{bmatrix}, \\\noalign{\medskip} \begin{bmatrix}  \frac{2}{s+6} & \frac{1}{s+6} & \frac{1}{s+6} & \frac{1}{s+6}  \end{bmatrix},\\\noalign{\medskip} \begin{bmatrix}  \frac{1}{s+6} & \frac{2}{s+6} & \frac{2}{s+6} & \frac{1}{s+6}\end{bmatrix} \end{array}  \] respectively.

The LFT representation of ${\cal S}_2(s)$ is represented as the LFT ${\cal F}(N_2,S_2) = $ where \[ N_2 = \left[ \begin{array}{c|c} \bf{0}_{4\times 1} & \bf{I_4} \\\hline \begin{array}{c} 0\\0 \\ 1\\0\\0\\1\end{array} & \begin{array}{cccc} 0 & 0 & 1&0\\0&0&0&1 \\ 0&0&0&0 \\ 1&0&0&0\\0&1&0&0\\0&0&0&0 \end{array}\end{array}\right] \] and $S_2(s), S_{21}(s),S_{22}(s)$ respectively given as \[\begin{array}{c}\begin{bmatrix}S_{21}(s) & 0  \\ 0 & S_{22}(s) \end{bmatrix} ,  \\\noalign{\medskip}  \begin{bmatrix} \frac{2(s^2+18s+76)}{s^3+21s^2+130s+234} & \frac{s^2+18s+76}{s^3+21s^2+130s+234} & \frac{s^2+19s+86}{s^3+21s^2+130s+234} \\ \frac{2(s^2+15s+52)}{s^3+21s^2+130s+234} & \frac{s^2+15s+52}{s^3+21s^2+130s+234} & \frac{(13+s)(s+5)}{s^3+21s^2+130s+234}\end{bmatrix} ,\\\noalign{\medskip} \begin{bmatrix} \frac{2s+13}{s^2+12s+34} & \frac{s+8}{s^2+12s+34} & \frac{7+s}{s^2+12s+34} \\ \frac{s+10}{s^2+12s+34} & \frac{2(7+s)}{s^2+12s+34} & \frac{s+8}{s^2+12s+34} \end{bmatrix}. \end{array}\]  

However, if we consider the minimal intricacy realizations of ${\cal C}_1, {\cal C}_2$ we get the same state-space realization $(A_o,B_o,C_o,D_o)$ with \[ \begin{array}{cc}A_o = \begin{bmatrix} -4 &1&2&1&1 \\1&-7&2&1&3\\2&1&-6&1&0\\1&2&2&-6&0\\1&2&0&0&-10 \end{bmatrix},& B_o = \begin{bmatrix} 1 \\ 1\\ 1\\ 1\\ 1 \end{bmatrix} \\\noalign{\medskip} C = \begin{bmatrix} I_4 & \bf{0}_{4\times 1}\end{bmatrix} \end{array} \]  The signal structure of the system is thus specified by the dynamical structure function $(Q,P)(s),$ with \[ \begin{array}{c} Q(s) = \begin{bmatrix} 0 & \frac{12+s}{s^2+14s+39} & \frac{2(s+10)}{s^2+14s+39} & \frac{s+10}{s^2+14s+39} \\ \frac{13+s}{s^2+17s+64} & 0 & \frac{2(s+10)}{s^2+17s+64} & \frac{s+10}{s^2+17s+64} \\ \frac{2}{s+6} & \frac{1}{s+6} & 0 & \frac{1}{s+6} \\ \frac{1}{s+6} & \frac{2}{s+6} &\frac{2}{s+6} & 0 \end{bmatrix}  \\\noalign{\medskip} P(s) = \begin{bmatrix} \frac{11+s}{s^2+14s+39} \\ \frac{13+s}{s^2+17s+64} \\ \frac{1}{s+6} \\ \frac{1}{s+6} \end{bmatrix} \end{array}\]
\end{example}

\noindent \begin{example}{\em A Subsystem Structure consistent with two Signal Structures} \\
Now we consider a subsystem structure with multiple signal structures.  Recalling the discussion above, subsystem structure describes the {\em closed loop} causal dependencies between manifest interconnection variables and signal structure specifies direct causal dependencies between manifest variables.  When it is impossible to determine these direct causal dependencies by inspection from the closed loop causal dependencies in a subsystem structure representation, then there can be multiple signal structures that are consistent with the subsystem structure.  

Reconsider ${\cal S}_2$.   The LFT is given by ${\cal F}(N_2,S_2) $ where \[ N_2 = \left[ \begin{array}{c|c} \bf{0}_{4\times 1} & \bf{I_4} \\\hline \begin{array}{c} 0\\0 \\ 1\\0\\0\\1\end{array} & \begin{array}{cccc} 0 & 0 & 1&0\\0&0&0&1 \\ 0&0&0&0 \\ 1&0&0&0\\0&1&0&0\\0&0&0&0 \end{array}\end{array}\right] \] and \[S_2 = \begin{bmatrix}S_{21} & 0  \\ 0 & S_{22} \end{bmatrix},\] with $S_{21}$ and $S_{22}$ given by  \[ \begin{array}{c}  \begin{bmatrix} \frac{2(s^2+18s+76)}{s^3+21s^2+130s+234} & \frac{s^2+18s+76}{s^3+21s^2+130s+234} & \frac{s^2+19s+86}{s^3+21s^2+130s+234} \\ \frac{2(s^2+15s+52)}{s^3+21s^2+130s+234} & \frac{s^2+15s+52}{s^3+21s^2+130s+234} & \frac{(13+s)(s+5)}{s^3+21s^2+130s+234}\end{bmatrix} ,\\\noalign{\medskip}\begin{bmatrix} \frac{2s+13}{s^2+12s+34} & \frac{s+8}{s^2+12s+34} & \frac{7+s}{s^2+12s+34} \\ \frac{s+10}{s^2+12s+34} & \frac{2(7+s)}{s^2+12s+34} & \frac{s+8}{s^2+12s+34} \end{bmatrix} \end{array}\] respectively. 
  
If we consider the relation  \[ S_2\left[\begin{array}{cc} K & L \end{array} \right] = (I-Q_{int})^{-1}\left[\begin{array}{cc} Q_{ext} & P \end{array} \right],\] we can take $(Q,P)(s)$ to equal  \[ \begin{array}{c} Q_1(s) = \begin{bmatrix} 0 & \frac{12+s}{s^2+14s+39} & \frac{2(s+10)}{s^2+14s+39} & \frac{s+10}{s^2+14s+39} \\ \frac{13+s}{s^2+17s+64} & 0 & \frac{2(s+10)}{s^2+17s+64} & \frac{s+10}{s^2+17s+64} \\ \frac{2}{s+6} & \frac{1}{s+6} & 0 & \frac{1}{s+6} \\ \frac{1}{s+6} & \frac{2}{s+6} & \frac{2}{s+6} & 0 \end{bmatrix}  \\\noalign{\medskip} P_1(s) = \begin{bmatrix} \frac{11+s}{s^2+14s+39} \\ \frac{13+s}{s^2+17s+64} \\ \frac{1}{s+6} \\ \frac{1}{s+6} \end{bmatrix} \end{array} \] with \[Q_{int} \equiv \left[ \begin {array}{cccc}  0 &{\frac { 12 +s}{{s}^{2}+ 14  s+ 39 }}& 0 & 0 \\\noalign{\medskip}{\frac { 13 +s}{{s}^{2}+ 17  s+ 64 }}& 0 & 0 & 0 \\\noalign{\medskip} 0 & 0 & 0 & \frac{1}{s+6}\\\noalign{\medskip} 0 & 0 & \frac{2}{s+6}& 0 \end {array} \right] \] and $Q_{ext} \equiv Q_1-Q_{int},$
or $Q_2(s)=$ \[   \begin{array}{c}\left[ \begin {array}{cccc}  0& 0&  {\frac {2({s}^{2}+ 18 s+ 76)}{{s}^{3}+ 21 {s}^{2}+ 130  s+ 234 }}&{\frac {{s}^{2}+ 18s+ 76 }{{s}^{3}+ 21  {s}^{2}+ 130s+ 234 }}\\\noalign{\medskip} 0& 0&  {\frac {2(52 + 15s+s^{2})}{{s}^{3}+ 21{s}^{2}+ 130s+ 234 }}&{\frac {52 + 15s+{s}^{2}}{{s}^{3}+ 21s^{2}+ 130s+ 234 }}\\\noalign{\medskip}{\frac { 2s+ 13 }{{s}^{2}+ 12s+ 34 }}&{\frac {s+ 8 }{{s}^{2}+ 12s+ 34 }}& 0 & 0 \\\noalign{\medskip}{\frac {s+ 10 }{{s}^{2}+ 12s+ 34 }}&   {\frac {2(7 +s)}{{s}^{2}+ 12s+ 34 }}& 0 & 0 \end {array} \right]   \\\noalign{\medskip} 
P_2(s) =\left[\begin{array}{c}
{\frac {{s}^{2}+ 19  s+ 86 }{{s}^{3}+ 21  {s}^{2}+ 130  s+ 234 }}\\\noalign{\medskip} 
{\frac { \left(  13 +s \right)  \left( s+ 5  \right) }{{s}^{3}+ 21  {s}^{2}+ 130  s+ 234 }}\\\noalign{\medskip} 
{\frac { 7 +s}{{s}^{2}+ 12  s+ 34 }}\\\noalign{\medskip} 
{\frac {s+ 8 }{{s}^{2}+ 12  s+ 34 }} \end{array} \right] \end{array} \] and taking $Q_{int} \equiv \left[\bf{0} \right]$ and $Q_{ext} \equiv Q_2.$  It is routine to show that with these definitions, both $(Q_1,P_1)(s)$ and $(Q_2,P_2)(s)$ are consistent with $S_2.$  Thus a single subsystem structure can be consistent with two signal structures.  It is important to note that one of our signal structures exhibited direct causal dependencies corresponding exactly with the closed loop dependencies described by subsystem structure. We note that this correspondence sometimes occurs (as in the case of Theorem \ref{SpecialMISO}) but generally does not hold.
\end{example}

\section{Impact on Systems Theory}

The introduction of partial structure representations suggests new problems in systems theory.  These problems explore the relationships between different representations of the same system, thereby characterizing properties of the different representations.  For example, classical realization theory considers the situation where a system is specified by a given transfer function, and it explores how to construct a consistent state space description.  Many important ideas emerge from the analysis:
\begin{enumerate}
\item State realizations are generally more informative than a transfer function representation of a system, as there are typically many state realizations consistent with the same transfer function,
\item Order of the state realization is a sensible measure of complexity of the state representation, and there is a well-defined {\em minimal} order of any realization consistent with a given transfer function; this minimal order is equal to the Smith-McMillian degree of the transfer function,
\item Ideas of controllability and observability of a state realization characterize important properties of the realization, and any minimal realization is both controllable and observable.
\end{enumerate}
In a similar way, introducing partial structure representations impacts a variety of concepts in systems theory, including realization, minimality, and model reduction.

\subsection{Realization}
The definition of partial structure representations enrich the kinds of realization questions one may consider.  In the previous section, we demonstrated that partial structure representations of a system are generally more informative than its transfer function but less informative than its state realization.  Thus, two classes of representation questions emerge: reconstruction problems and structure realization problems (Figure \ref{fig:problems}).

\begin{figure}[thb]
\centering
  \includegraphics*[width=.45\textwidth]{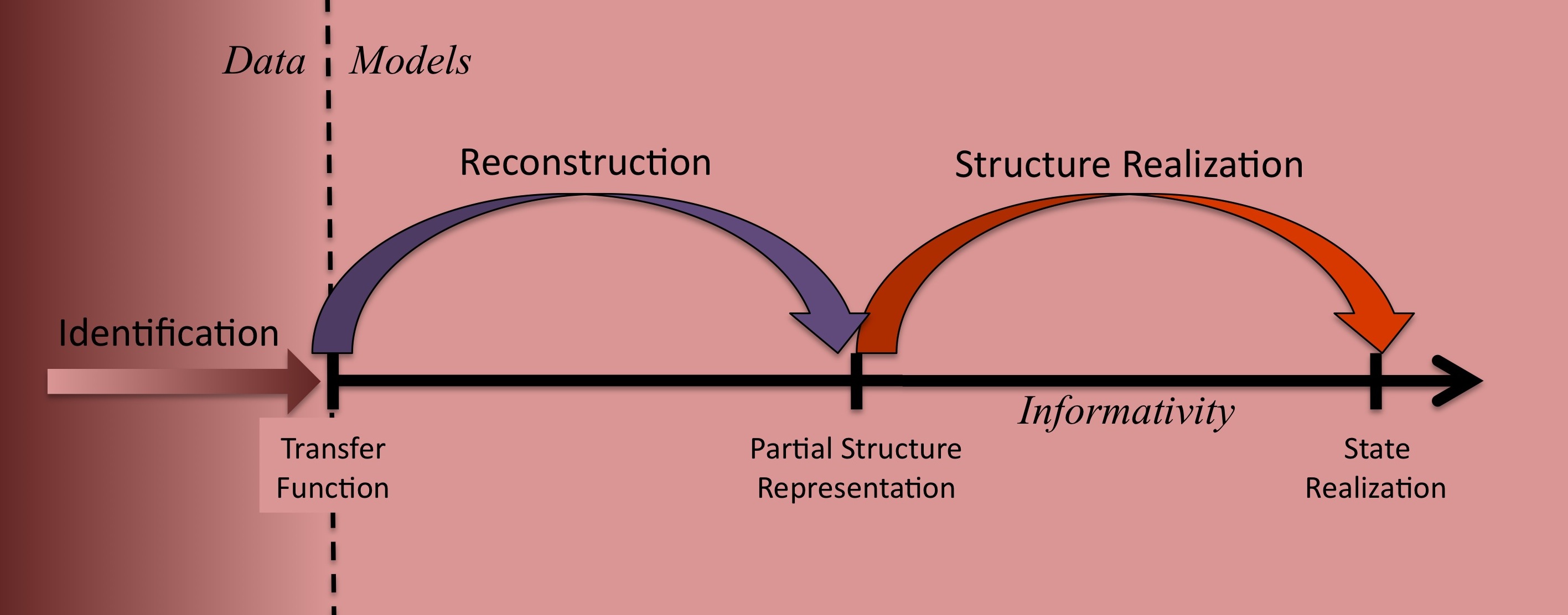}
\caption{Partial structure representations introduce new classes of realization problems: reconstruction and structure realization.  These problems are distinct from identification, and each depends on the type of partial structure representation considered.} 
 \label{fig:problems}     
\end{figure}

Reconstruction problems consider the construction of a partial structure representation of a system given its transfer function.  Because partial structure representations are generally more informative than a transfer function, these problems are ill-posed--like the classical realization problem.  Nevertheless, one may consider algorithms for generating canonical representations, or one may characterize the additional information about a system--beyond knowledge of its transfer function--that one would require in order to specify one of its partial structure representations.  In particular, we may consider two types of reconstruction problems:
\begin{enumerate}
\item {\bf Signal Structure Reconstruction:} Given a transfer function $G(s)$ with its associated sparsity structure $\cal Z$, find a signal structure, $\cal W$, with dynamical structure function $(Q, P)$ as in Equation (\ref{eq:DSF}) such that $G=(I-Q)^{-1}P$,
\item {\bf Subsystem Structure Reconstruction:} Given a transfer function $G(s)$ with its associated sparsity structure $\cal Z$, find a subsystem structure, $\cal S$, with structured LFT $(N,S)$ as in Equation (\ref{eq:LFT}) such that $G = (I-SK)^{-1}SL.$
\end{enumerate}
Signal structure reconstruction is also called {\em network reconstruction}, particularly in systems biology where it plays a central role.  There, the objective is to measure fluctuations of various proteins, or other chemical species, in response to particular perturbations of a biochemical system, and then infer causal dependencies among these species \cite{ourBook,Rangel2005,yuancdc09,ourTrans,yuanMTNS10,yuanbook09}. 

Structure realization problems then consider the construction of a state space model, possibly generalized to include auxiliary variables as necessary, consistent with a given partial structure representation of a system.   Like the classical realization problem or reconstruction problems, these problems are also ill-posed since there are typically many state realizations of a given partial structure representation of a system.  Also, like reconstruction problems, these realization problems can be categorized in two distinct classes, depending on the type of partial structure representation that is given: 
\begin{enumerate}
\setcounter{enumi}{2}
\item {\bf Signal Structure Realization:} Given a system $G$ with signal structure $\cal W$ and associated dynamical structure function $(Q,P)$, find a state space model $(A,B,C,D)$ consistent with $(Q,P)$, i.e. such that Equations (\ref{eq:laplace}-\ref{eq:DSF}) hold,
\item {\bf Subsystem Structure Realization:} Given a system $G$ with subsystem structure $\cal S$ and associated structured $(N,S)$ (recorded in the LFT representation ${\cal F}(N,S)$), find a generalized state space model of the form of Equation (\ref{eq:linearsystem}) consistent with $(N,S)$.
\end{enumerate}
Signal structure realization may sometimes be called {\em network realization}, consistent with the nomenclature for signal structure reconstruction.

Note that all the reconstruction and structure realization problems here are distinct from identification problems, just as classical realization is different from identification.  For the systems considered here, identification refers to the use of input-output data (and no other information about a system) to choose a representation that best describes the data in some appropriate sense.  Because input-output data only characterizes the input-output map of a system, identification can at best characterize the system's transfer function; no information about structure, beyond the sparsity structure, is available in such data.  In spite of this distinction, however, it is not uncommon for reconstruction problems to be called {\em structure identification} problems.  Nevertheless, one may expose such problems as the concatenation of an identification and a reconstruction problem and precisely characterize the extra information needed to identify such structure by carefully analyzing the relevant reconstruction problem, independent of the particular identification technique \cite{ourTrans,ourBook}.

\subsection{Minimality}
Just as partial structure representations enrich the classical realization problem, they also impact the way we think about minimality.  Certainly the idea of a minimal complexity representation is relevant for each of the four problems listed above, but clearly the relevant notion of complexity may be different depending on the representation.  We consider each situation as follows:
\begin{enumerate}
\item {\bf Minimal Signal Structure Reconstruction:}  In this situation one needs to consider how to measure the complexity of a system's signal structure, $\cal W$, or its associated dynamical structure function, $(Q, P)$.  Some choices one may consider could be the number of edges in $\cal W$, the maximal order of any element in $(Q, P)$, or the maximal order of any path from any input $u_i$ to any output $y_j$.  The problem would then be to find a minimal complexity signal structure consistent with a given transfer function.
\item {\bf Minimal Subsystem Structure Reconstruction:} In this situation one needs to consider how to measure the complexity of a system's subsystem structure, $\cal S$, or its associated structured LFT, $(N,S)$.  One notion could be to measure complexity by the number of distinct subsystems; the problem would then be to find the minimal complexity subsystem structure representation consistent with a given transfer function.  Another notion could be the number of non-zero entries in the $S$ matrix, where $(N,S)$ denote the the LFT associated with the subsystem structure.  Using this measure, a single subsystem block with no zero entries would be considered a more complex representation than a subsystem structure with a large number of distinct, albeit interconnected, subsystems. 
 \item {\bf Minimal Signal Structure Realization:} In this situation one needs to consider how to measure the complexity of a system's zero-intricacy state realization, from which signal structure is derived.  The obvious choice would be to use the order of the realization as a natural measure of complexity, and the problem would then be to find the minimal order state realization \cite{yuancdc09,yuanITAC11} consistent with a given signal structure, $\cal W$, or, equivalently, with its associated dynamical structure function, $(Q,P)$.  Note that this minimal order is guaranteed to be finite (for the systems considered here) and can easily be shown to be greater or equal to the Smith-McMillian degree of the transfer function specified by the signal structure; we call this number the {\em structural degree} of the signal structure \cite{yeungcdc}. 
\item {\bf Minimal Subsystem Structure Realization:} In this situation one needs to consider how to measure the complexity of a generalized state realization in the presence of auxiliary variables.  Both the order and the intricacy of the realization offer some perspective of its complexity, but one needs to consider how they each contribute to the overall complexity of the realization.  The problem would then be to find a minimal complexity generalized state realization consistent with a given subsystem structure.  
\end{enumerate}  

These various problems demand new ideas for thinking about the complexity of a system's representation, especially that of a partial structure representation.  These new ideas about complexity, in turn, introduce opportunities for characterizing minimality of a representation in terms that add insight to our understanding of the relationship between a system's behavior and its structure, much as controllability and observability characterize classical notions of minimality in a system's state realization.  Besides suggesting the need for a characterization of minimality, however, these ideas also impact notions of approximation and how we think about model reduction.

\subsection{Model Reduction}
Each of the reconstruction and structure realization problems described above have associated with them not only a minimal-representation version of the problem, but also an approximate-representation version of the problem.  The minimal-representation versions of these problems, as described above, seek to construct  a representation of minimal complexity in the targeted class that is nevertheless consistent with the system description provided.  Similarly, approximate-representation versions of these problems seek to construct a representation in the targeted class that has a lower complexity than the minimal complexity necessary to deliver consistency with the system description provided.  As a result, consistency with the given system description can not be achieved, so measures of approximation become necessary to sensibly discuss a ``best" representation of the specified complexity.  

For example, associated with the classical realization problem is the standard model reduction problem.  In this situation, a transfer function is specified, and one would like to construct a state realization with a complexity that is lower than that which is necessary (for such a realization to be consistent with the given transfer function) that nevertheless ``best" approximates it.  Note that the appropriate notion of complexity depends on the target representation class; here, the target representation class is the set of state realizations, so the appropriate notion of complexity is model order.  Likewise, note that the appropriate notion of approximation depends on the type of system representation that is initially provided; here, a transfer function is provided, so an appropriate measure of approximation could be an induced norm, such as $H_{\infty}$.  Thus, one could measure the quality of an approximation by measuring the induced norm of the error between the given transfer function and that specified by the approximate state realization.  Note that since the $H_{\infty}$ model reduction problem remains open, many alternative measures and approaches have been suggested \cite{dullerudpaganini,megretski06,ZDG96}.  In any event, because the specified system description is a transfer function, the resulting measure of approximation is typically one that either directly or indirectly measures the difference in input-output dynamic behavior between the approximate model and the given system description; the focus is on dynamics, not system structure, when considering notions of approximation in the standard model reduction problem.  

Similarly, there is an appropriate reduction problem associated with each of the minimal reconstruction and minimal structure realization problems described above.  Like standard model reduction, the appropriate notion of complexity is characterized by the class of target representations, while the appropriate measure of approximation depends on the system representation that is initially provided.  To make these ideas more concrete, we discuss each of the problems individually:
\begin{enumerate}
\item {\bf Approximate Signal Structure Reconstruction:}  In this situation one would like to find a signal structure representation with lower complexity e.g. fewer number of edges, etc. than the minimal level of complexity necessary to specify a given transfer function.   When using the number of edges as a complexity measure, this problem may be interpreted as trying to restrict the number of communication links between individual computation nodes in a distributed system while approximating some desired global dynamic.  The appropriate measure of approximation, then, is any measure that sensibly compares the desired transfer function from that specified by the reduced signal structure, such as $H_{\infty}$. 
\item {\bf Approximate Subsystem Structure Reconstruction:}  In this situation one would like to find a subsystem structure representation with lower complexity than the minimal level of complexity necessary to specify a given transfer function.  If one considers the number of subsystems as the measure of complexity, then this problem is trivial since any transfer function is consistent with a subsystem structure with a single subsystem.  If one measures complexity by the number of non-zero entries in the $S$ matrix, where $(N,S)$ denote the the LFT associated with the subsystem structure, then it appears likely that one could often formulate a meaningful approximation problem.
\end{enumerate}

Unlike these approximate reconstruction problems, approximate realization problems may have dual relevant measures of approximation, since having an initial partial structure representation of a system also always specifies its transfer function.  Measuring the similarity between transfer functions determines the degree to which a lower order system approximates the {\em dynamics} of a given system, while measuring the similarity between partial structure representations determines the degree to which a lower order system approximates the {\em structure} of a given system.     

\begin{enumerate}
\setcounter{enumi}{2}
\item {\bf Approximate Signal Structure Realization:}  In this situation one would like to find a state realization with a model complexity that is lower than the minimal complexity necessary to specify a given signal structure.  Typically the complexity of a generalized state realization would be measured by both the order and intricacy of the realization.  Since signal structure only depends on the zero intricacy realization of a system, however, a minimal complexity realization will always have zero intricacy; thus model order is the only relevant notion of complexity.  Moreover, since the structural degree of a particular signal structure may be strictly greater than the Smith-McMillian degree of the transfer function it specifies, a few distinct kinds of approximation problems emerge: structural approximation and dual approximation (Figure \ref{fig:assr}).
\begin{enumerate}
\item {\bf Structural Approximation:}  When the order of the approximation is less than the structural degree of the given signal structure but greater than the Smith-McMillian degree of the associated transfer function, the resulting approximation can specify the given transfer function exactly, even though its signal structure only approximates that of the given system.  Note that the sense in which similarity in signal structure should be measured is an area of on-going research.  
\item {\bf Dual Approximation:} When the order of the approximation is less than the Smith-McMillian degree of the transfer function specified by the given signal structure, then it will  represent neither the structure nor the dynamics of the given system exactly.
 \end{enumerate} 
\begin{figure}[htb]
\centering
  \includegraphics*[width=.45\textwidth]{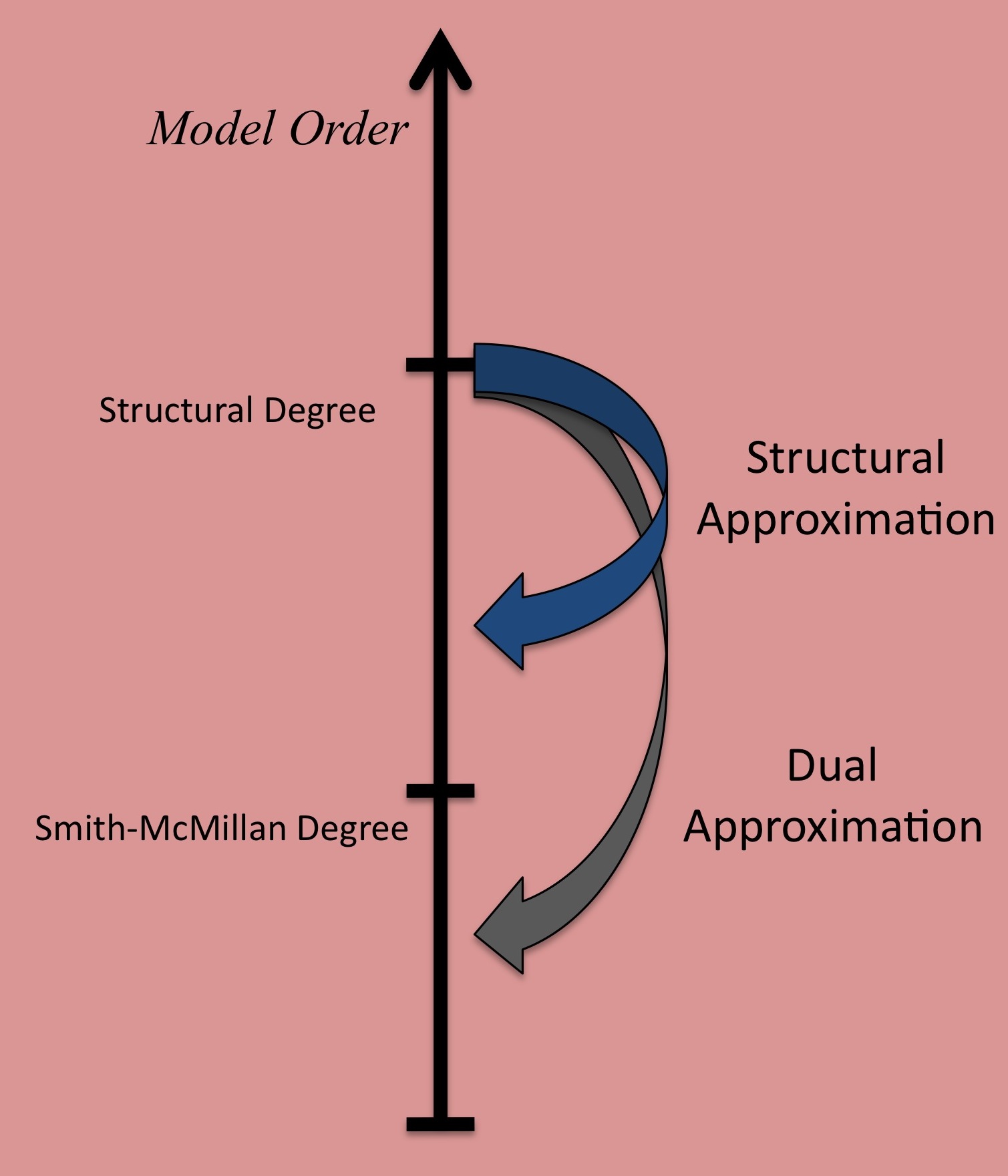}
\caption{Approximate signal structure realization leads to two distinct types of reduction problems.  Structural approximation exactly realizes the dynamics of a given signal structure while only approximating its structure; dual approximation captures neither the dynamics nor the structure of a given signal structure exactly.} 
 \label{fig:assr}     
\end{figure}
\item {\bf Approximate Subsystem Structure Realization:} In this situation one would like to find a generalized state realization with a model complexity that is lower than the minimal complexity necessary to specify a given subsystem structure.  Here, complexity of a generalized state realization is measured both in terms of intricacy and order since intricacy of a realization directly impacts the number of admissible subsystem blocks (see Figure \ref{fig:linear_example}) while order impacts the ability of the realization to approximate the transfer function specified by the given subsystem structure.  As a result, three distinct approximation problems emerge to complement subsystem structure realization:  Structure-Preserving Model Reduction, Subsystem Structure Approximation, Subsystem Dual Approximation (Figure \ref{fig:mr}).
\begin{enumerate}
\item {\bf Structure-Preserving Model Reduction: } When the intricacy of an approximation is high enough, the subsystem structure of a system can be preserved while its dynamic behavior is approximated by lower-order systems.  A naive approach to such reduction would be to reduce the order of various subsystems.  However, even when each subsystem is well approximated, the dynamic behavior of the overall interconnection may be very different from the original system.  As a result, methods have been developed in this area for approximating the dynamics of the closed-loop, interconnected system while preserving its subsystem structure \cite{sandberg+07b,san+08a,san+08b}.    
\item {\bf Subsystem Structure Approximation:}   When the order of an approximation is high enough, the dynamic behavior of a system can be preserved while is subsystem structure is approximated by lower-intricacy realizations.  The sense in which similarity of subsystem structure should be measured remains an open topic for research.
\item {\bf Subsystem Dual Approximation:}  When both the order and the intricacy of an approximation are lower than the minimal values necessary to realize a given subsystem structure and the transfer function it specifies, then the objective of the reduction problem is to find the realization of the specified complexity that best approximates the structure and dynamics of the given system.
\end{enumerate}   
\end{enumerate}

\begin{figure}[htb]
\centering
  \includegraphics*[width=.45\textwidth]{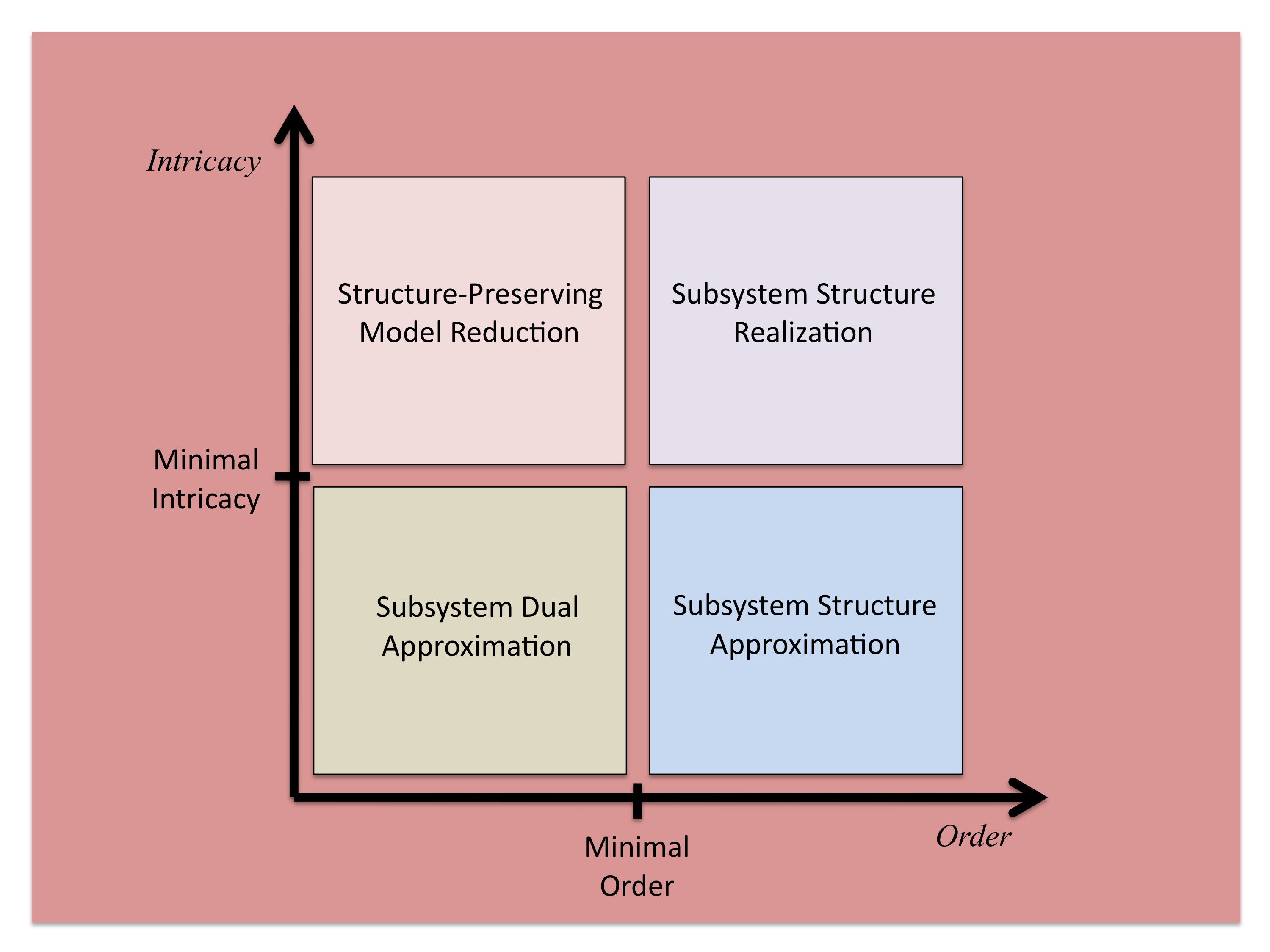}
\caption{Approximate subsystem structure realization leads to three distinct types of reduction problems: 1) Structure-preserving model reduction preserves the subsystem structure of a given system while approximating its dynamic behavior, 2) Subsystem structure approximation preserves the dynamic behavior but approximates the subsystem structure, and 3) Subsystem dual approximation captures neither the structure nor the dynamic behavior of a system described by a particular subsystem structure.} 
 \label{fig:mr}     
\end{figure}

The introduction of partial structure representations suggests a number of new problems in systems theory.  These problems include new classes of realization problems, called reconstruction and structural realization problems, as well as a number of new reduction problems.  Each of these problems differ depending on the partial structure representation one considers, and a number of research issues remain to properly formulate most of them.  The overview offered here is merely meant to give a perspective of the landscape of problems that emerges with the introduction of partial structure representations.

 \section{Conclusion}

This paper introduced the idea of a system's complete computational structure as a baseline for comparing simplified representations.  Although closely aligned with a system's state space realization, we demonstrated the need for auxiliary variables to encode information about a system's admissible subsystem structure.  This results in a ``generalized" state description that is a differential algebraic system with differentiation index zero, and the number of auxiliary variables becomes an additional measure of complexity (besides model order) that we call {\em intricacy}.  This generalized state representation, and its associated graphical representation, is the system's complete computational structure.

Partial structure representations were then introduced as a means for simplifying a system's structural description while retaining a complete representation of its input-output dynamic behavior.  These included subsystem structure, signal structure, and sparsity structure.

Subsystem structure was first introduced as the most refined view of the interconnection of a system's legitimate subsystems.  This description is represented as the lower linear fractional transformation of a static interconnection matrix $N$ with a block diagonal operator $S$.  Its graphical representation is a block diagram, like the system's complete computational structure, that is a condensation graph with respect to a meaningful partition of the system's states. 

Signal structure, on the other hand, is a signal flow diagram of the causal dependencies among manifest variables given by the dynamical structure function of the system.  We demonstrated that systems may exhibit extremely structured behavior in the signal structure sense while having no apparent structure in the subsystem or computational structure sense.  Moreover, the transformations between manifest variables, represented as edges in the signal structure graph, do not necessarily partition the system states, as do the nodes of the subsystem structure.  This fact implies that the minimal order to realize a particular signal structure may be, in fact, higher than the minimal order necessary to realize the transfer function specified by the given signal structure.

Finally, the weakest notion of structure is sparsity structure, or the pattern of zero entries in the system's transfer function matrix, which graphically also becomes a signal flow graph like signal structure.  We demonstrated that this representation is very weak, reminding readers that a diagonal transfer function matrix, for example, does not imply that even a minimal realization of the system is necessarily decoupled.  Thus, this notion of structure really is a statement about the closed-loop dependencies of inputs on outputs of the system and much less of a statement about how those closed-loop dependencies came to be.

These representations were then shown to contain differing levels of structural information about a system.  In particular, it was shown that the complete computational structure uniquely specifies both the subsystem and signal structure of a system, and that either of these partial structure representations uniquely specify the transfer function (and thus its associated sparsity structure) of the system.  Nevertheless, the relationship between subsystem and signal structure is less definitive, as we demonstrated that two realizations of the same system may share subsystem structure but have different signal structures, or, conversely, two realization of the same system may share signal structure but have different subsystem structures.  These different representations simply appear to offer different perspectives of the system's structural properties.

We then surveyed the landscape of new problems in systems theory that arise when considering these partial structure representations.  In particular, we showed that the standard realization problem now becomes two types of problems, a reconstruction problem, where a transfer function is given and one would like to determine a partial structure representation that is consistent with it, and a structure realization problem, where a partial structure representation is given and one would like to find a generalized state realizations that is consistent with it.  Minimal versions of these problems are obtained if one can define a sensible notion of complexity of each kind of representation, and various suggestions were offered.  Associated approximation, or model reduction problems were then characterized, where the target representation is simpler than the minimal complexity necessary to yield a representation that is consistent with the given model or description of the system.  Here we note that the model reduction problems begin to consider structure approximation as well as approximation of the system's dynamic behavior, leading to a variety of new problems one may consider.  It is our hope that  many of these problems will be addressed in the coming years, leading to a more thorough and complete understanding of the meaning of structure and its relationship to the behavior of interconnected dynamic systems.

\section{Acknowledgments}
This work was supported in part by a Charles Lee Powell Foundation Fellowship, the Air Force Research Laboratory Grant FA 8750-09-2-0219 and by the Engineering and Physical Sciences Research Council Grant EP/G066477/1.
\bibliography{../../library}

\begin{thebibliography}{}

\end{thebibliography}


\begin{thebibliography}{1}

\bibitem{Mortveit01}
H.S. Mortveit and C.M. Reidys.
\newblock Discrete, sequential dynamical systems.
\newblock {\em Discrete Mathematics}, 226:281--295, 2001.

\bibitem{Mortveit2000}
H.S. Mortveit and C.M. Reidys.
\newblock Elements of a theory of simulation ii: Sequential dynamical systems.
\newblock {\em Applied Mathematics and Computation}, 107:121--136, 2001.

\end{thebibliography}


\begin{thebibliography}{10}

\bibitem{nyquist1932}
H.~Nyquist,
\newblock ``Regeneration theory'',
\newblock {\em Bell System Technical Journal}, vol. 11, pp. 126--147, September
  1932.

\bibitem{black1934}
H.~S. Black,
\newblock ``Stabilized feedback amplifiers'',
\newblock {\em Bell System Technical Journal}, vol. 13, pp. 1--18, 1934.

\bibitem{bode1940}
H.~W. Bode,
\newblock ``Relations between attenuation and phase in feedback amplifier
  design'',
\newblock {\em Bell System Technical Journal}, vol. 19, pp. 421--454, 1940.

\bibitem{maxwell1868}
J.~C. Maxwell,
\newblock ``On governors'',
\newblock {\em Proceedings of the Royal Society}, vol. 16, no. 100, pp.
  270--283, 1868.

\bibitem{RazaCSM96}
H.~Raza and P.~Ioannou,
\newblock ``Vehicle following control design for automated highway systems'',
\newblock {\em Control Systems Magazine}, vol. 16, no. 6, December 1996.

\bibitem{FowlerDAndreaCSM03}
J.~Fowler and R.~D'Andrea,
\newblock ``A formation flight experiment'',
\newblock {\em Control Systems Magazine}, October 2003.

\bibitem{DAndrea}
R.~D'Andrea and G.~Dullerud,
\newblock ``Distributed control for spatially interconnected systems'',
\newblock {\em IEEE Transactions on Automatic Control}, vol. 48, no. 9, pp.
  1478--1496, September 2003.

\bibitem{DelVecchio08}
D.~Del Vecchio, A.~J. Ninfa, and E.~D. Sontag,
\newblock ``Modular cell biology: retroactivity and insulation'',
\newblock {\em Molecular Systems Biology}, vol. 4, no. 161, 2008.

\bibitem{DelVecchio09}
S.~Jayanthi and D.~Del Vecchio,
\newblock ``On the compromise between retroactivity attenuation and noise
  amplification in gene regulatory networks'',
\newblock {\em Proceedings of the 2009 IEEE Conference on Decision and
  Control}, December 2009.

\bibitem{luenberger78}
D.~Luenberger,
\newblock ``Time-invariant descriptor systems'',
\newblock {\em Automatica}, vol. 14, pp. 473--480, 1978.

\bibitem{luenberger}
David~G. Luenberger,
\newblock {\em Introduction to Dynamic Systems : Theory, Models, and
  Applications},
\newblock Wiley, N.Y., 1992.

\bibitem{luenberger77}
D.~Luenberger,
\newblock ``Dynamic equations in descriptor form'',
\newblock {\em IEEE Transactions on Automatic Control}, vol. 22, no. 3, pp.
  312--321, June 1977.

\bibitem{schmidt}
K.~Schmidt,
\newblock {\em Dynamical Systems of Algebraic Origin},
\newblock Basel, Boston, 1943.

\bibitem{willems07}
J.~C. Willems,
\newblock ``The behavioral approach to open and interconnected systems'',
\newblock {\em Control Systems Magazine}, vol. 27, no. 6, pp. 46--99, 2007.

\bibitem{Siljac}
D.~Siljac,
\newblock {\em Large Scale Dynamic Systems: Stability and Structure},
\newblock New York: North-Holland, 1978.

\bibitem{}
S.~Warnick and J.~Gon\c{c}alves,
\newblock ``Necessary and sufficient conditions for dynamical structure
  reconstruction'',
\newblock {\em IEEE Transactions of Automatic Control}, August 2008.

\bibitem{murrayCON}
J.~Fax R.~Olfati-Saber and R.~Murray,
\newblock ``Consensus and cooperation in networked multi-agent systems'',
\newblock {\em Proceedings of the IEEE}, vol. 95, no. 1, pp. 215--233, January
  2007.

\bibitem{roy09}
S.~Roy, Y.~Wan, and A.~Saberi,
\newblock ``On time scale design for networks'',
\newblock {\em International Journal of Control}, vol. 49, no. 9, pp.
  1313--1325, July 2009.

\bibitem{Parlangeli}
G.~Parlangeli and G.~Notarstefano,
\newblock ``On the observability of path and cycle graphs'',
\newblock {\em 49th IEEE Conference on Decision and Control}, December 2010.

\bibitem{Mortveit01}
H.S. Mortveit and C.M. Reidys,
\newblock ``Discrete, sequential dynamical systems'',
\newblock {\em Discrete Mathematics}, vol. 226, pp. 281--295, 2001.

\bibitem{Mortveit2000}
H.S. Mortveit and C.M. Reidys,
\newblock ``Elements of a theory of simulation ii: Sequential dynamical
  systems'',
\newblock {\em Applied Mathematics and Computation}, vol. 107, pp. 121--136,
  2001.

\bibitem{harary}
F.~Harary,
\newblock {\em Graph Theory},
\newblock Addison-Wesley Pub. Co., Reading, Massachusetts, 1969.

\bibitem{ZDG96}
K.~Zhou, J.C. Doyle, and K.~Glover,
\newblock {\em Robust and Optimal Control},
\newblock Prentice Hall, Upper Saddle River, New Jersey, 1996.

\bibitem{ourTrans}
J.~Gon\c{c}alves and S.~Warnick,
\newblock ``Dynamical structure functions for the reverse engineering of {LTI}
  networks'',
\newblock {\em IEEE Transactions of Automatic Control, 2007}, August 2007.

\bibitem{yeungcdc}
E.~Yeung, J.~Gon\c{c}alves, H.~Sandberg, and S.~Warnick,
\newblock ``Network structure preserving model reduction with weak a priori
  structural information'',
\newblock {\em Proceedings of 2009 Conference on Decision and Control},
  December 2009.

\bibitem{yeungCDC10}
E.~Yeung, J.~Goncalves, H.~Sandberg, and S.~Warnick,
\newblock ``Representing structure in linear interconnected dynamical
  systems'',
\newblock {\em The Proceedings of the 49th IEEE Conference on Decision and
  Control to appear}, December 2010.

\bibitem{howescdc08}
R.~Howes, S.~Warnick, and J.~Gon\c{c}alves,
\newblock ``Dynamical structure analysis of sparsity and minimality heuristics
  for reconstruction of biochemical networks'',
\newblock {\em Proceedings of the Conference on Decision and Control}, December
  2008.

\bibitem{ourCDC}
J.~Gon\c{c}alves, R.~Howes, and S.~Warnick,
\newblock ``Dynamical structure functions for the reverse engineering of {LTI}
  networks'',
\newblock {\em Proceedings of the Conference on Decision and Control}, 2007.

\bibitem{yuancdc09}
Yuan Y., Stan G.B., Warnick S., and Gon\c{c}alves J.,
\newblock ``Minimal dynamical structure realisations with application to
  network reconstruction from data'',
\newblock in {\em Proceedings of the 48th IEEE Conference on Decision and
  Control (IEEE-CDC 2009)}, December 2009.

\bibitem{fosbe09}
C.~Ward, E.~Yeung, T.~Brown, B.~Durtschi, S.~Weyerman, R.~Howes, J.~Goncalves,
  H.~Sandberg, and S.~Warnick.,
\newblock ``A comparison of network reconstruction methods for chemical
  reaction networks'',
\newblock {\em Proceedings of the Foundations for Systems Biology and
  Engineering Conference}, August 2009.

\bibitem{ourBook}
S.~Warnick and J.~Gon\c{c}alves,
\newblock {\em Systems Theoretic Approaches to Network Reconstruction},
  chapter~13, pp. 265--296,
\newblock MIT Press, Cambridge, Massachusetts, 2009.

\bibitem{Rangel2005}
Claudia Rangel, John Angus, Zoubin Ghahramani, and David~L. Wild,
\newblock {\em Modeling Genetic Regulatory Networks using Gene Expression
  Profiling and State-Space Models}, chapter~9, pp. 269--293,
\newblock Springer, London, UK, 2005.

\bibitem{yuanMTNS10}
Warnick~S. Yuan~Y., Stan~G. and Goncalves J.,
\newblock ``Robust dynamical network reconstruction'',
\newblock in {\em 19th International Symposium on Mathematical Theory of
  Networks and Systems (MTNS 2010), 2010}, 2010.

\bibitem{yuanbook09}
Stan~G. Yuan~Y. and Goncalves J.,
\newblock ``Biological network reconstruction from data. (for engineer)'',
\newblock in {\em Engineering principle in Biology}. Cold Spring Habour Lab,
  2009.

\bibitem{yuanITAC11}
S.~Warnick Yuan~Y., G.~Stan and J.~Gonalves,
\newblock ``Minimal realisation of dynamical structure function and its
  application to network reconstruction'',
\newblock in {\em IEEE Transactions on Automatic Control, to appear}, 2011.

\bibitem{dullerudpaganini}
G.E. Dullerud and F.~Paganini,
\newblock {\em A course in robust control theory --- a convex approach},
\newblock Springer-Verlag, 2000.

\bibitem{megretski06}
A.~Megretski,
\newblock ``H-infinity model reduction with guaranteed suboptimality bound'',
\newblock in {\em American Control Conference, 2006}, june 2006, p. 6 pp.

\bibitem{sandberg+07b}
Henrik Sandberg and Richard~M. Murray,
\newblock ``Frequency-weighted model reduction with applications to structured
  models'',
\newblock in {\em Proceedings of the 2007 American Control Conference}, New
  York City, New York, July 2007, pp. 941--946.

\bibitem{san+08a}
Henrik Sandberg and Richard~M. Murray,
\newblock ``Model reduction of interconnected linear systems'',
\newblock {\em Optimal Control; Application and Methods, Special Issue on
  Directions, Applications, and Methods in Robust Control}, 2009.

\bibitem{san+08b}
Henrik Sandberg and Richard~M. Murray,
\newblock ``Model reduction of interconnected linear systems using structured
  gramians'',
\newblock in {\em Proceedings of the 17th IFAC World Congress}, Seoul, Korea,
  July 2008, pp. 8725--8730.

\end{thebibliography}
\end{document}